\documentclass[conference, letterpaper,twocolumn]{IEEEtran}
\usepackage{epsfig}
\usepackage{times}
\usepackage{float}
\usepackage{afterpage}
\usepackage{amsmath}
\usepackage{amstext}
\usepackage{amssymb,bm, bbm}
\usepackage{latexsym}
\usepackage{color}
\usepackage{graphicx}
\usepackage{amsmath}
\usepackage{amsthm}
\usepackage{graphicx}
\usepackage[center]{caption}
\usepackage{subfigure}
\usepackage{booktabs}
\usepackage{multicol}
\usepackage{lipsum}
\usepackage{dblfloatfix}
\usepackage{mathrsfs}
\usepackage{cite}
\usepackage{tikz}
\usepackage{pgfplots}
\pgfplotsset{compat=newest}
 \usepackage{enumitem}
\usepackage{makecell}
\usepackage{float}
\restylefloat{table}

\allowdisplaybreaks

\newcommand{\expect}[1]{{\mathbb{E}}\left[#1\right]}
\newcommand{\expcnd}[2]{{\mathbb{E}}\left[ #1 \;\middle|\; #2\right]}
\newcommand{\expcndn}[3]{{\mathbb{E}^{#1}}\left[ #2 \;\middle|\; #3\right]}

\newcommand{\lr}[1]{\left( #1 \right)}

\newcommand{\rmd}{\mathrm{d}}
\newcommand{\bbE}{\mathbb{E}}\newcommand{\rme}{\mathrm{e}}

\newcommand{\bbN}{\mathbb{N}}

\newcommand{\bbP}{\mathbb{P}}

\newcommand{\bbR}{\mathbb{R}}

\newcommand{\sfA}{\mathsf{A}}
\newcommand{\sfB}{\mathsf{B}}

\newcommand{\bfe}{\mathbf{e}}

\newcommand{\sfN}{\mathsf{N}}

\newcommand{\bfp}{\mathbf{p}}

\newcommand{\cA}{\mathcal{A}}
\newcommand{\cB}{\mathcal{B}}

\newcommand{\cD}{\mathcal{D}}

\newcommand{\cP}{\mathcal{P}}

\newcommand{\cX}{\mathcal{X}}

\newcommand{\del}{\partial}

\newcommand{\D}{D}
\newcommand{\kl}[2]{{\D}\left(\left.#1 \, \right\| #2 \right)}

\newcommand{\supp}{{\mathsf{supp}}}

\theoremstyle{mystyle}
\newtheorem{theorem}{Theorem}
\theoremstyle{mystyle}
\newtheorem{lemma}{Lemma}
\theoremstyle{mystyle}
\newtheorem{prop}{Proposition}
\theoremstyle{mystyle}
\newtheorem{corollary}{Corollary}
\theoremstyle{mystyle}
\theoremstyle{remark}
\theoremstyle{mystyle}
\theoremstyle{mystyle}
\theoremstyle{mystyle}
\theoremstyle{discussion}
\theoremstyle{mystyle}
\theoremstyle{mystyle}

 \usepackage{enumitem}

\begin{document}

\title{Binomial Channel: On the Capacity-Achieving Distribution and Bounds on the Capacity   }

\author{
\IEEEauthorblockN{Ian Zieder$^{*}$, Antonino Favano$^{\dagger}$, Luca Barletta$^{\dagger}$, and Alex Dytso$^{**}$}
$^{*}$ New Jersey Institute of Technology, Newark, NJ 07102, USA. Email: ihz3@njit.edu \\
$^{\dagger}$ Politecnico di Milano, Milano, 20133, Italy. Email: $\{$antonino.favano, luca.barletta$\}$@polimi.it \\
$^{**}$ Qualcomm Flarion Technologies, Bridgewater,  NJ 08807, USA.
Email: odytso2@gmail.com 
}

\maketitle

\begin{abstract}
This work considers a binomial noise channel. The paper can be roughly divided into two parts. The first part is concerned with the properties of the capacity-achieving distribution. In particular, for the binomial channel, it is not known if the capacity-achieving distribution is unique since the output space is finite (i.e., supported on integers $0, \ldots, n)$ and the input space is infinite (i.e., supported on the interval $[0,1]$), and there are multiple distributions that induce the same output distribution. This paper shows that the capacity-achieving distribution is unique by appealing to the total positivity property of the binomial kernel.   In addition, we provide upper and lower bounds on the cardinality of the support of the capacity-achieving distribution. Specifically, an upper bound of order   $ \frac{n}{2}$ is shown, which improves on the previous upper bound of order $n$ due to Witsenhausen. Moreover, a lower bound of order $\sqrt{n}$ is shown. 
Finally, additional information about the locations and probability values of the support points is established. 

The second part of the paper focuses on deriving upper and lower bounds on capacity. In particular, firm bounds are established for all $n$ that show that the capacity scales as $\frac{1}{2} \log(n)$.
 \end{abstract}

\section{Introduction}

We consider a  channel for which the relationship between the input $X \in [0,1]$ and the output $Y \in  \{0,\ldots, n \}$ is described by the binomial distribution:
\begin{equation}
    P_{Y|X}(y|x) = {n \choose y} x^y (1-x)^{n-y}. 
\end{equation}

In this work, we are interested in studying the capacity of this channel as a function of the number of trials $n$, that is
\begin{equation}\label{eq:capacity_opt_problem}
    C(n) = \max_{ P_X: \,  X \in [0,1] } I(X;Y). 
\end{equation}
In addition to studying capacity, we are also interested in studying properties of an optimal capacity-achieving distribution distribution denoted by $P_{X^\star}$. 

\subsection{Literature Review}
\label{sec:Lit}
The binomial channel naturally arises in molecular communications and the interested reader is referred to \cite{farsad2020capacities,einolghozati2013design,farsad2017capacity,jamali2019channel} and references therein. The channel is also useful in the study of the deletion channel \cite{levenshtein1966binary,cheraghchi2019capacity}.  

The capacity of the binomial channel was first considered in \cite{komninakis2001capacity} where the authors used minimax redundancy theorem in \cite{xie1997minimax} to argue that asymptotically the capacity scales as $\frac{1}{2} \log n $. 
The exact capacity for the $n=1$ case was computed in \cite{farsad2020capacities} where binary distribution with support on $\{0,1 \}$ was shown to be capacity-achieving.  To the best of our knowledge, there are no firm  bounds on the capacity of the binomial channel.

Properties of the capacity-achieving distribution have also been looked at. For example, the authors of \cite{farsad2020capacities} have designed an algorithm for computing capacity and a capacity-achieving distribution by using a dual representation of the maximization problem. It is also known that, by using the Witsenhausen technique \cite{WitsenhausenBOund}, there exists a capacity-achieving distribution with at most $n+1$ mass points. We note, however, that the Witsenhausen technique does not guarantee that the optimal input distribution is unique. In fact, for the binomial channel, uniqueness has not been shown; note that uniqueness is important not just for theoretical purposes, but also for algorithmic purposes.  A conventional way to show that the capacity-achieving distribution is unique is by establishing that the mutual information is a \emph{strictly} concave function of the input distribution. However, as will be shown by an example, for the binomial channel, the mutual information is not strictly concave. 
Other properties, such as location of the support points, are also not well understood.  The main goal of this work is to close some of these gaps. 

In this work, we also rely on estimation theoretic quantities such as the conditional expectation. For the estimation theoretic treatments of the binomial channel, the interested reader is referred to \cite{taborda2014information,BD_Low}. Recently, deterministic identification capacity for the binomial channel has been studied in \cite{salariseddigh2023deterministic}. 

\subsection{Outline and Contributions}
The paper outline and contributions are as follows. The remaining part of Sec.~I is dedicated to notation.  Sec.~\ref{sec:Preliminaries} presents the required preliminary and ancillary results. In particular, Sec.~\ref{sec:KKT} presents the KKT conditions and some small yet important consequences of these conditions, and Sec.~\ref{sec:EstPrelim} establishes properties of estimation theoretic quantities, such as the conditional mean, that will be needed in our analysis.  
Sec.~\ref{sec:main_I} constitutes the first part of our main results and focuses on properties of capacity-achieving distributions.  In particular, in Sec.~\ref{sec:symmetry} it is shown that all capacity-achieving distributions are symmetric around $\frac{1}{2}$; in Sec.~\ref{sec:discretness}  it is shown that all capacity-achieving distributions are discrete.  In Sec.~\ref{sec:uniqness}, the discreetness is used to argue strong concavity of the mutual information, which implies uniqueness of $P_{X^\star}$.   Sec.~\ref{sec:location} provides additional information about the location of the support points.  Sec.~\ref{sec:probabilities} provides bounds on the probability values. Sec.~\ref{sec:card_bounds} improves the upper bound $n+1$ on cardinality of the support, due to Witsenhausen, to the bound of order $\frac{n}{2}$, and  provides a lower bound of order~$\sqrt{n}$. Sec.~\ref{sec:cap_bounds} constitutes the second part of our main results and focuses on bounds on the capacity.  In particular, firm lower and upper bounds of order $\frac{1}{2} \log(n)$ are derived.

\subsection{Notation}

All logarithms are to the base $\rme$. Deterministic scalar quantities are denoted by lower-case letters and random variables are denoted by uppercase letters.  
For a random variable $X$ and every measurable subset $\cA \subseteq \bbR$ the probability distribution is written as $P_{X}(\cA) = \bbP[X \in \cA]$. The support set of $P_X$ is
\begin{equation}
\supp(P_{X})=\{x:  \text{ $P_{X}( \mathcal{D})>0$ for every open set $ \mathcal{D} \ni x $ } \}. 
\end{equation} 
When $X$ is discrete, we write $P_X(x)$ for $P_X(\{x\})$, i.e., $P_X$ is a probability mass function (pmf). The relative entropy of the distributions $P$ and $Q$ is $\kl{P}{Q}$. 

Given a function $f: \bbR \to \bbR$ and a set $\cB \subseteq \bbR$, the number of zeros of $f$ in $\cB$ is given by
\begin{equation}
\sfN\left(\cB; f\right) = \left|     \{x : f(x) =0 \} \cap \cB \right |,  
\end{equation}
where $| \cdot |$ denotes the cardinality. 

The set of the first $n$ positive integers is denoted by $[n]$. The entry in position $(i,j)$ of matrix $\sfA$ is denoted by $[\sfA]_{i,j}$.

\section {Preliminaries}
\label{sec:Preliminaries}
We now presents some of the  tools needed in our analysis. 

\subsection{KKT Conditions }
\label{sec:KKT}

The key that allows one to study properties of the support  of an optimal input distribution is the following lemma which contains the KKT conditions for our optimization problem~\cite{CISS_2018}.

\begin{lemma}\label{lem:KKT}
     $P_{X^\star}$ is a capacity-achieving input distribution if and only if the following conditions hold:  
\begin{align}
    i(x;P_{Y^\star}) &\le C(n), \qquad x \in [0,1], \\
    i(x;P_{Y^\star}) &= C(n), \qquad x \in \supp(P_{X^\star}) \label{eq:KKT_equality}
\end{align}
where   $P_{X^\star} \to P_{Y|X} \to P_{Y^\star}$ (i.e., the optimal output distribution) and 
\begin{equation}
    i(x;P_{Y^\star})=\kl{P_{Y|X}(\cdot|x)}{P_{Y^\star}}. 
  \end{equation}
  
\end{lemma}


We also define the following set, which would be useful in our study of the uniqueness of $P_{X^\star}$:
\begin{equation}\label{eq:zeros_KKT}
\cA_n =\{x \in [0,1]:\:  i(x;P_{Y^\star}) -C(n)=0 \} .
\end{equation}
The importance of  $\cA_n$ is demonstrated in the following lemma. 
\begin{lemma}\label{lem:Uniquness_A_n} For a given $n$
    \begin{itemize}
    \item $\cA_n$ is unique; and 
    \item  $\supp(P_{X^\star} )  \subseteq \cA_n$ for every $P_{X^\star}$. 
    \end{itemize}
\end{lemma}
\begin{proof}
    Note that, for a given $n$, both $P_{Y^\star}$ and $C(n)$ are unique (even if $P_{X^\star}$ is not unique) \cite{kemperman1974shannon} and, since $\cA_n$ only depends on these quantities, the uniqueness follows. 

    The second part follows from the KKT conditions in Lemma~\ref{lem:KKT}, because $x \in \supp(P_{X^\star} )$ implies $x \in \cA_n$. 
\end{proof}


\subsection{Estimation Theoretic Preliminaries}
\label{sec:EstPrelim}
Estimation theoretic quantities will play an important role in our analysis. In what follows, the quantity $\expcndn{n-1}{f(Y)}{X=x}$ denotes expectation with respect to a binomial distribution with $n-1$ trials and success probability $x$ per trial, and 
\begin{equation}
    \ell_b(x,\hat{x}) = x \log\left(\frac{x (1-\hat{x})}{(1-x)\hat{x}}\right)-\frac{x-\hat{x}}{1-\hat{x}}, \quad (x,\hat{x}) \in (0,1)^2
\end{equation}
represents the Bregman divergence for the binomial channel.

We now summarize some of these preliminary results. 

\begin{prop}\label{prop:derivatives_info_density} For $n\ge 2$ and $x \in (0,1)$, we have
\begin{align}
    i'(x;P_Y)  &= \frac{n}{x} \expcndn{n-1}{\ell_b(x,\expcndn{n-1}{X}{Y})}{X=x} \nonumber\\
    &\quad +\frac{n}{x} \expcndn{n-1}{\frac{x-\expcndn{n-1}{X}{Y}}{1-\expcndn{n-1}{X}{Y}}}{X=x}\label{eq:inf_den_derivative_with_bregman_1}
\end{align}
and
\begin{equation}
i''(x;P_Y) = \frac{n}{x(1-x)}+\frac{1}{(1-x)^2}  G(x) \label{eq:i_second_der} 
\end{equation}
where $G(x)$ is defined in \eqref{eq:G_Definition} (at the top of next page).
\begin{figure*}
\begin{equation}
G(x) = \expcnd{(n-Y)(n-Y-1)\log\frac{\expcnd{X}{Y=Y}}{\expcnd{1-X}{Y=Y+1}}\frac{\expcnd{1-X}{Y=Y+2}}{\expcnd{X}{Y=Y+1}}}{X=x} \label{eq:G_Definition}
\end{equation}
\vspace{-0.4cm}
\end{figure*}
\end{prop}
\begin{IEEEproof}
See Appendix~\ref{sec:Derivatives_info_density}. 
\end{IEEEproof}

The Bregman divergence in \eqref{eq:inf_den_derivative_with_bregman_1} appeared previously in a different but related result, specifically in 
\cite[Prop.~2]{taborda2014information} it was shown that  for $a \in (0,1)$
\begin{equation}
    \frac{\partial}{ \partial a}  I(  X; \cB_n(a X) ) = \frac{n}{a} \bbE \left[ \ell_b \left(aX, \bbE[aX|\cB_{n-1}(a X')  ] \right) \right] \label{eq:derivative:idenity}
\end{equation}
where $Y=\cB_n(aX)$ denotes the transformation of input $aX$ through a binomial channel with $n$ trials.  

Finally, we also need to show the monotonicity of the conditional mean. 

\begin{lemma}\label{lem:cond_exp_non_increasing}
    The function $y \mapsto \expcnd{X}{Y=y}$ is non-decreasing.
\end{lemma}
\begin{proof}
    First of all, note that
    \begin{equation} \label{eq:Tweedy}
        \expcnd{X}{Y=y} = \frac{\expect{X^{y+1}(1-X)^{n-y}}}{\expect{X^{y}(1-X)^{n-y}}}.
    \end{equation}
    Let us now introduce the functions $f_1, f_2, g_1, g_2$ as follows:
    \begin{align}
        &f_1(x) = x^{y}, \quad f_2(x) = x^{y+1}, \\
        &g_1(x) = (1-x)^{n-y}, \quad g_2(x) = x(1-x)^{n-y-1},
    \end{align}
    and note that the functions 
    \begin{equation}
        \frac{f_2(x)}{f_1(x)} = x, \qquad \frac{g_2(x)}{g_1(x)} = \frac{x}{1-x}
    \end{equation}
    are both increasing and non-negative for $x \in [0,1]$. As a consequence, the entries and the determinant of the matrices 
    \begin{equation}
        \left[\begin{array}{cc}
            f_1(x_1) & f_1(x_2) \\
            f_2(x_1) & f_2(x_2)
        \end{array}\right], \qquad \left[\begin{array}{cc}
            g_1(x_1) & g_1(x_2) \\
            g_2(x_1) & g_2(x_2)
        \end{array}\right],
    \end{equation}
    are non-negative for any choice of $0\le x_1 < x_2 \le 1$. By using the \emph{basic composition formula} of~\cite[Ch.~3.1]{karlin1968total}, we can also say that the entries and the determinant of the matrix
    \begin{equation}
        \left[\begin{array}{cc}
            \expect{f_1(X)g_1(X)} & \expect{f_1(X)g_2(X)} \\
            \expect{f_2(X)g_1(X)} & \expect{f_2(X)g_2(X)}
        \end{array}\right]
    \end{equation}
    are non-negative. Therefore, we have
    \begin{equation}
        \frac{\expect{f_2(X)g_2(X)}}{\expect{f_1(X)g_2(X)}} \ge \frac{\expect{f_2(X)g_1(X)}}{\expect{f_1(X)g_1(X)}}
    \end{equation}
    or
    \begin{equation}
        \frac{\expect{X^{y+2} (1-X)^{n-y-1}}}{\expect{X^{y+1}(1-X)^{n-y-1}}} \ge \frac{\expect{X^{y+1}(1-X)^{n-y}}}{\expect{X^y (1-X)^{n-y}}},
    \end{equation}
    which, by using \eqref{eq:Tweedy}, is the same as 
    \begin{equation}
        \expcnd{X}{Y=y+1} \ge \expcnd{X}{Y=y}.
    \end{equation}
    This concludes the proof. 
\end{proof}

\section{Properties of the Capacity-Achieving Distributions}
\label{sec:main_I}
In this section we study properties of capacity-achieving distributions. 

\subsection{Symmetry} 
\label{sec:symmetry}

The binomial channel exhibits the following symmetry
\begin{equation} \label{eq:channel_symmetry}
    P_{Y|X}(y|x) = P_{Y|X}(n-y|1-x), \quad x \in [0,1], \, y \in \{0\}\cup [n]. 
\end{equation}
which  immediately leads to to the following result.  

\begin{prop}\label{prop:input_symmetry}
If $X^\star$ is capacity-achieving, then  $X^\star \stackrel{d}{=}1-X^\star$.\footnote{Here $\stackrel{d}{=}$ denotes equality in distribution.} 
\end{prop}

\subsection{Discreteness}
\label{sec:discretness}
As already mentioned in Section~\ref{sec:Lit}, from the Witsenhausen approach we only know that there exists a discrete distribution with at most $n+1$ mass points.  This, however, does not rule out the existence of other capacity-achieving distributions (\emph{e.g.},  continuous capacity-achieving distributions).  




We now show that all capacity-achieving distributions are discrete and provide  a preliminary bound on the support. 
\begin{prop}\label{prop:bound_cardinality_cA} $| \cA_n| \le n+1$.
\end{prop}
This bound will be improved in Section~\ref{sec:card_bounds}. 
\subsection{Uniqueness of the Optimal Input Distribution} 
\label{sec:uniqness}

In this section, we show and discuss uniqueness of the capacity-achieving input distribution. To aid our discussion, it is useful to parameterize the mutual information in terms of distributions instead of random variables, that is 
\begin{equation}
    I(P_X; P_{Y|X} ) = I(X;Y). 
\end{equation}
We also let $\cP_{\cX}$ be the set of all distributions over the set $\cX$. In particular, the optimization in \eqref{eq:capacity_opt_problem} can be written as
\begin{equation}
   \max_{P_X \in \cP_{[0,1]}} I(P_X; P_{Y|X} ).
\end{equation}
A typical way to show that there is a unique maximizer  is to show that the mapping $P_X \mapsto I(P_X; P_{Y|X} ) $ over the set $\cP_{[0,1]}$ is \emph{strictly} concave \cite{smith1971information}.  However, due to the fact that the output space of the binomial channel is finite and the input space is uncountable, the mutual information is not strictly concave over $\cP_{[0,1]}$.  For example, when $n=1$ any  distribution symmetric around $x=\frac{1}{2}$ will induce 
\begin{equation}
    P_{Y}(0) = P_{Y}(1) = \frac{1}{2} 
\end{equation}
which is the capacity-achieving output distribution for $n=1$.  Therefore, to show uniqueness of the capacity-achieving input distribution a new or slightly different argument is needed. 

 We begin by showing the following result. 

\begin{prop}\label{prop:full_rank_A}  Consider an arbitrary sequence $0\le x_1 < \ldots  < x_{n+1} \le 1$ and define the
matrix $\sfA \in \bbR^{n+1 \times n+1} $ as
\begin{equation}
[\sfA]_{i,k} = P_{Y|X}(i-1|x_k), \quad  i \in [n+1], \, k \in [n+1]. 
\end{equation}
Then,  $\sfA$ is full rank. 
\end{prop}
\begin{proof}
First of all, we argue that considering $x_1=0$ and $x_{n+1}=1$ comes without loosing generality. In fact, in this case
the first and last columns of $\sfA$ are $\bfe_1$ and $\bfe_{n+1}$, respectively, where $\bfe_i$ is a zero vector with a $1$ in the $i$-th position. As a consequence, we have $\det(\sfA) = \det(\tilde{\sfA})$, where
\begin{equation}
 [\tilde{\sfA}]_{i,k} = [\sfA]_{i+1,k+1}, \quad i \in [n-1], \, k \in [n-1].   
\end{equation}  
     Next, note that we can rewrite the binomial law as 
    \begin{equation}
        P_{Y|X}(y|x) = \binom{n}{y}(1+\rme^{t})^{-n}\rme^{ty} 
    \end{equation}
    where $x = \frac{\rme^t}{1+\rme^t}$. The matrix $\sfB$ with $[\sfB]_{y,k} = \rme^{t_k y}$ and $y \in [n-1]$ is a Vandermonde matrix, which is full rank since the $t_k$'s are all distinct \cite{golub2013matrix}. Thanks to the multilinear property of the determinant, we can write that 
    \begin{equation}
        \det(\tilde{\sfA}) = \det(\sfB) \prod_{y=1}^{n-1} \binom{n}{y} \prod_{k=2}^{n} (1+\rme^{t_k})^{-n} >0
    \end{equation}
    where the last step is due to $\det(\sfB)>0$ and to the positivity of the products. As a consequence, $\sfA$ is a full rank matrix.
\end{proof}

With  the aid of  Proposition~\ref{prop:full_rank_A},  we show the following result. 
\begin{theorem}\label{thm:stric:_concavity}
    Let $\cX \subset [0,1]$ be a discrete set of cardinality $n+1$.  Then, $P_X \mapsto I(P_X; P_{Y|X})$ is \emph{strictly} concave over $\cP_{\cX}$. 
\end{theorem}
\begin{proof} 
    Let $P_X, Q_X \in \cP_{\cX}$, and let $P_X^\epsilon = (1-\epsilon) P_X +\epsilon Q_X$ for $\epsilon \in (0,1)$, which is also in $\cP_{\cX}$.  Moreover, 
    let $P_X \to P_{Y|X} \to P_Y$,  $Q_X \to P_{Y|X} \to Q_Y$ and $P_X^\epsilon \to P_{Y|X} \to P_Y^\epsilon$. 
    Then, first note that
    \begin{align}
    &I( P_X^\epsilon; P_{Y|X} ) \notag\\
    &- (1-\epsilon)  I(  P_X ; P_{Y|X} ) - \epsilon I(   Q_X; P_{Y|X} )\\
    &=     D (P_{Y|X} \| P_{Y}^\epsilon| P_X^\epsilon)  \notag\\
    &-(1-\epsilon)  D (P_{Y|X} \| P_{Y}| P_X)  - \epsilon  D (P_{Y|X} \| Q_{Y}| Q_X) \\
     &= (1-\epsilon) D(P_{Y}\| P_{Y}^\epsilon) +\epsilon  D(Q_{Y}\| P_{Y}^\epsilon). \label{eq:concavity_decomposition}
    \end{align}

    We now show that every $ P_X \in \cP_{\cX}$ induces a distinct output distribution (i.e., $P_X \to P_{Y|X} \to P_Y$ is an injective mapping), which implies that \eqref{eq:concavity_decomposition} is strictly positive and, therefore, the mutual information is strictly concave.  Define the following:
    \begin{align}
    \bfp_X &= [  P_X(x_1), \ldots, P_X(x_{n+1})],  \quad x_k \in \cX, \\
    \bfp_Y &= [P_Y(0), \ldots, P_Y(n) ].
    \end{align}
    Then, the mapping $P_X \to P_{Y|X} \to P_{Y}$ can be written as the following system of linear equations:
    \begin{equation}
    \sfA \bfp_X = \bfp_Y \label{eq:linear_system}
    \end{equation}
    where the matrix $\sfA \in \bbR^{n+1 \times n+1}$ is such that
\begin{equation}
    [\sfA]_{i,k} = P_{Y|X}(i-1|x_k), \,  i \in [n+1], \, x_k \in \cX. 
\end{equation}
    From Proposition~\ref{prop:full_rank_A}, we know that $\sfA$ is full rank for any $\cX$ of cardinality $n+1$. Therefore, from standard linear algebra result, it follows that the mapping in \eqref{eq:linear_system} is injective (i.e., every $\bfp_X$ induces a distinct $\bfp_Y$).  Therefore, we conclude that \eqref{eq:concavity_decomposition} is positive and mutual information is strictly concave.  
\end{proof}

Note that since  by Proposition~\ref{prop:bound_cardinality_cA}, $\cA_n$ has cardinality of at most $n+1$, from Theorem~\ref{thm:stric:_concavity} we have the following corollary.  

\begin{corollary}
    $P_X \mapsto I(P_X; P_{Y|X})$ is strictly concave over $\cP_{\cA_n}$. Consequently, $P_{X^*}$ is unique.   
\end{corollary}





\subsection{On the Location of Support Points}
\label{sec:location}

 Following the same lines of~\cite[Sec.~V]{abou2001capacity} we have that: 
\begin{prop}\label{prop:0_is_optimal}
    Let $P_{X^\star}$ be a capacity-achieving input distribution. Then, $\{0,1\} \in \supp(P_{X^\star})$.
\end{prop}
\begin{proof}
    By using symmetry (Prop.~\ref{prop:input_symmetry}), we can just prove the result for the point at $x=0$. Let $0\le x_0 < x_1  < \ldots < x_N \le 1$ be the support points of $P_X$. Suppose that $x_0>0$. Then, we have that
    \begin{align}
        &\frac{\del}{\del x_0}I(X;Y) = P_X(x_0) \frac{\del}{\del x_0} \expcnd{\log\frac{P_{Y|X}(Y|x_0)}{P_Y(Y)}}{X=x_0} \\
        &= \frac{P_X(x_0)}{x_0(1-x_0)} \expcnd{(Y-nx_0)\log\frac{P_{Y|X}(Y|x_0)}{P_Y(Y)}}{X=x_0}. \label{eq:der_mut_inf}
    \end{align}
    Next, we prove that the function $f: y\mapsto \log\frac{P_{Y|X}(y|x_0)}{P_Y(y)}$ is decreasing. Note that
    \begin{align}
        &\frac{P_Y(y)}{P_{Y|X}(y|x_0)} = P_{X}(x_0) + \sum_{i=1}^N P_{X}(x_i) \left(\frac{x_i}{x_0}\right)^y \left(\frac{1-x_i}{1-x_0}\right)^{n-y} \\
        &= P_{X}(x_0) + \sum_{i=1}^N P_{X}(x_i) \left(\frac{\frac{1}{x_0}-1}{\frac{1}{x_i}-1}\right)^y \left(\frac{1-x_i}{1-x_0}\right)^{n}
    \end{align}
    is an increasing function of $y$, since $x_0 < x_i$ for $i\ge1$. As a consequence, the function $f$ is decreasing. By noting that $\expcnd{Y}{X=x_0} = n x_0$ and by applying \cite[Lemma 1]{abou2001capacity} to \eqref{eq:der_mut_inf}, we get that $\frac{\del}{\del x_0}I(X;Y)<0$ for all $0<x_0<x_1$. This implies that $x_0 = 0 \in \supp(P_{X^\star})$. 
\end{proof}


An important consequence of Proposition~\ref{prop:0_is_optimal} is given next. 
\begin{corollary}\label{cor:rel_cap_probY}
The channel capacity is equal to
\begin{equation}
    C(n) = \log\frac{1}{P_{Y^\star}(0)} = \log\frac{1}{P_{Y^\star}(n)}.
\end{equation}
\end{corollary}
\begin{proof}
    Thanks to Proposition \ref{prop:0_is_optimal}, we know that $0 \in \supp(P_{X^\star})$. By using the KKT condition \eqref{eq:KKT_equality}, we can write
    \begin{equation*}
        C(n) = i(0;P_{Y^\star}) = \sum_{y=0}^n {n \choose y} 0^y \log\frac{{n \choose y} 0^y}{P_{Y^\star}(y)} = \log\frac{1}{P_{Y^\star}(0)}.
    \end{equation*}
    By symmetry, we can argue that $P_{Y^\star}(0) = P_{Y^\star}(n)$. 
\end{proof}

We next show that there is at most one support point in the interval $\left(0,\frac{1}{n} \right]$ and, by symmetry, at most one point in $\left[1-\frac{1}{n},1 \right)$.  The proof technique we use was developed in \cite{mceliece1979practical} in the context of Poisson noise channels. 
\begin{prop} \label{prop:loc_info}
    For all $n\ge 1$, we have
    \begin{align}
        \left| \supp(P_{X^\star}) \cap \left(0,\frac{1}{n}\right] \right|&\le 1, \label{eq:limit_supp_1} \\
        \left| \supp(P_{X^\star}) \cap \left[1-\frac{1}{n},1\right) \right|&\le 1.\label{eq:limit_supp_2}
    \end{align}
\end{prop}
\begin{proof}
    For $n=1$, the claim follows from the fact that there are only two mass points at $\{0,1\}$. Next, we consider the case $n\ge 2$. 
    
    From expression~\eqref{eq:i_doubleprime} of Proposition~\ref{prop:der_info_den}, we have that $x(1-x)i''(x;P_{Y^\star})$ is given by  the expression in \eqref{eq:sec_der_inf_den_limited_supp}, where in the last step we have exploited the channel symmetry~\eqref{eq:channel_symmetry}.
    \begin{figure*}
    \begin{align}
        &x(1-x)i''(x;P_{Y^\star}) \\
        &= n+\frac{x}{1-x}  \expcnd{(n-Y)(n-Y-1)\log\frac{\expcnd{X^\star}{Y=Y}}{\expcnd{1-X^\star}{Y=Y+1}}\frac{\expcnd{1-X^\star}{Y=Y+2}}{\expcnd{X^\star}{Y=Y+1}}}{X=x} \\
        &= n+\sum_{y=0}^{n-2} \binom{n}{y} x^{y+1}(1-x)^{n-y-1} (n-y)(n-y-1)\log\frac{\expcnd{X^\star}{Y=y}}{\expcnd{X^\star}{Y=y+1}}\frac{\expcnd{X^\star}{Y=n-y-2}}{\expcnd{X^\star}{Y=n-y-1}} \label{eq:sec_der_inf_den_limited_supp}
    \end{align} 
    \end{figure*} 
    
    Since $y\mapsto \expcnd{X^\star}{Y=y}$ is a non-decreasing function (see Lemma~\ref{lem:cond_exp_non_increasing}), all the terms in the summation of~\eqref{eq:sec_der_inf_den_limited_supp} are non-positive. Moreover, the functions $x\mapsto x^{y+1}(1-x)^{n-y-1}$ for $y=0, \ldots, n-2$ are increasing for $x \le \frac{1}{n}$. As a consequence, the function $x\mapsto g(x)=x(1-x)i''(x;P_{Y^\star})$ is non-increasing for $x \in \left(0,\frac{1}{n}\right]$. Since $g(0)=n\ge 2$, the function $g$ has at most one zero in the interval $\left(0,\frac{1}{n}\right]$. Then, $i''$ has at most one zero in the interval $\left(0,\frac{1}{n}\right]$, hence $i(x;P_{Y^\star})-C(n)$ has at most one zero crossing in the interval $\left(0,\frac{1}{n}\right]$. This proves~\eqref{eq:limit_supp_1}. To prove~\eqref{eq:limit_supp_2}, we can use symmetry $X^\star \stackrel{d}{=} 1-X^\star$ from Proposition~\ref{prop:input_symmetry}.
\end{proof}

\subsection{Bounds on the Probabilities}
\label{sec:probabilities}

We begin by recalling that for $P_X \to P_{Y|X} \to P_Y$ and $Q_X  \to P_{Y|X} \to Q_{Y}$, we have that
\begin{equation}
    \kl{P_X}{Q_X} = \kl {P_Y}{Q_Y} + D (P_{X|Y} \| Q_{X|Y} |P_Y)  ,\label{eq:KL_identity}
\end{equation}
where the conditional relative entropy is defined as
\begin{equation}
    D (P_{X|Y} \| Q_{X|Y} |P_Y) = \sum_{y=0}^n P_Y(y) \kl{P_{X|Y}(\cdot|y)}{Q_{X|Y}(\cdot|y)}
\end{equation}

The key to finding bounds on the probabilities is the following lemma. 

\begin{lemma}\label{lem:expression_for_prob} For $x^\star \in \supp(P_{X^\star})$
    \begin{equation}
    P_{X^\star}(x^\star) = \rme^{ -C(n) -  \cD(x^\star) }, 
\end{equation}
where 
\begin{equation}
  \cD(x^\star) =  D  \left(\delta_{x^\star} \| P_{X^\star|Y} |P_{Y|X }(\cdot|x^\star) \right). 
\end{equation}
\end{lemma}
\begin{proof}
     Using the  equality condition  in the KKT \eqref{eq:KKT_equality}, we have that for $x^{\star} \in \supp(P_{X^{\star}})$
\begin{align}
   C(n) &= \kl{P_{Y|X}(\cdot|x^\star)}{P_{Y^\star}}\\
    &=  \kl{ P_{Y_{x^\star} }  }{P_{Y^\star}} \label{eq:Def_output_P_Y}\\
    &= \kl{\delta_{x^\star} }{  P_{X^\star} } -    D (\delta_{x^\star} \| P_{X^\star|Y} |P_{Y_{x^\star} }) \label{eq:apply_KL_identity}\\
    &= \log \frac{1}{P_{X^\star} (x^\star)} -    D (\delta_{x^\star} \| P_{X^\star|Y} |P_{Y_{x^\star} }), \label{eq:last_step_P(x)}
\end{align}
where  \eqref{eq:Def_output_P_Y} follows by defining $\delta_{x^\star} \to P_{Y|X} \to P_{Y_{x^*}}$;  and \eqref{eq:apply_KL_identity} follows by using \eqref{eq:KL_identity}.

By rearranging~\eqref{eq:last_step_P(x)}, and recognizing that $P_{Y_{x^*}}(\cdot) = P_{Y|X}(\cdot |x ^\star)$,  we arrive at: for $x^{\star} \in \supp(P_{X^{\star}})$ 
\begin{equation}
    P_{X^\star}(x^\star) = \rme^{ -C(n) -  D  \left(\delta_{x^\star} \| P_{X^\star|Y} |P_{Y|X }(\cdot|x^\star) \right) }. 
\end{equation}
\end{proof}

The term $\cD(x^\star)$ measures how on average the $P_{X^\star|Y}$ is close to a point measure. We refer to $\cD(x^\star)$ as the \emph{crest-factor}.\footnote{In signal processing, the crest factor  measures how peaky the waveform is. Specifically, it compares the peak amplitude of a waveform relative to its average value.}  

From Lemma~\ref{lem:expression_for_prob}, by  using $\cD(x^\star) \ge 0$, which follows from the non-negativity of the  relative-entropy, we immediately arrive at the following bound: 
\begin{equation}
    P_{X^\star}(x^\star) \le  \rme^{ -C(n) }, \qquad  x^\star \in \supp(P_{X^\star}). \label{eq:simple_prob_bound}
\end{equation}
The bound in \eqref{eq:simple_prob_bound} might appear ineffective due to the fact that the capacity is unknown.  However, note that for any $\tilde{X}$, from the definition of the capacity we have that 
\begin{equation}
    P_{X^\star}(x^\star) \le  \rme^{ - I(\tilde{X}; \tilde{Y}) },  \qquad  x^\star \in \supp(P_{X^\star}),
\end{equation}
which implies that any good guess results in an upper bound.

The next result improves upon the bound in \eqref{eq:simple_prob_bound}.

\begin{prop}
 \label{prop:crest_factor_bounds}
\text{ }
\begin{itemize}
    \item First Bound: for $x \in \supp(P_{X^\star}) \setminus \{0,1\} $
    \begin{align}\label{eq:second_bound_factor_crest}
      \cD(x) \ge  \frac{(1-x)^n \log(1-x)^n+x^n \log(x^n)}{(1-x)^n+x^n-1}
    \end{align}
    \item Second Bound: for $x\in \supp(P_{X^\star}) \setminus\{\frac{1}{2}\}$
    \begin{align}\label{eq:first_bound_factor_crest}
     \cD(x) \ge    \log  \left( 1 + \left( \frac{x}{1-x} \right)^{n(1-2x)} \right)
    \end{align}
\end{itemize}
    
\end{prop}
\begin{proof}
We begin by noticing that 
\begin{equation}
-\cD(x) = \expcnd{\log(P_{X^\star|Y}(x|Y))}{X=x} . \label{eq:Minus_D}
\end{equation}
 To show the first bound notice that for $x\in (0,1)$ 
     \begin{align}
     &\expcnd{\log(P_{X^\star|Y}(x|Y))}{X=x}  \notag\\
     &\le  \expcnd{(\mathbbm{1}(Y=0)+\mathbbm{1}(Y=n))\log(P_{X^\star|Y}(x|Y))}{X=x} \\
         &= P_{Y|X}(0|x) \log\frac{P_{Y|X}(0|x) P_{X^\star}(x)}{P_{Y^\star}(0)}  \notag\\
         & \quad+ P_{Y|X}(n|x) \log\frac{P_{Y|X}(n|x) P_{X^\star}(x)}{P_{Y^\star}(n)} \\
         &= (1-x)^n \left(\log((1-x)^n P_{X^\star}(x)) +C(n)\right)  \notag\\
         & \quad + x^n\left( \log(x^n P_{X^\star}(x))+C(n) \right) \label{eq:use_rel_capacity_probY}\\
         &= \left((1-x)^n+x^n\right)\left(\log(P_{X^\star}(x)) +C(n)\right)  \notag\\
         &\quad + (1-x)^n \log(1-x)^n+x^n \log(x^n) \label{eq:bound_exp_log_P_X_given_Y}\\
         &= - \left((1-x)^n+x^n\right) \cD(x)  \notag\\
         &\quad + (1-x)^n \log(1-x)^n+x^n \log(x^n), \label{eq:using_Def_0f_cD}
     \end{align}
     where in~\eqref{eq:use_rel_capacity_probY} we used Corollary~\ref{cor:rel_cap_probY}; and where in \eqref{eq:using_Def_0f_cD} we have used Lemma~\ref{lem:expression_for_prob}. Combing \eqref{eq:using_Def_0f_cD} with \eqref{eq:Minus_D} we arrive at the desired first bound.

    To show the second bound note that by symmetry of the optimal distribution,   we have that, for $x \neq 1/2$,
\begin{align}
&P_{X^*|Y}(x|y) \notag  \\
&= \frac{P_{X^\star}(x) P_{Y|X}(y|x) }{P_{Y^\star}(y)}\\
&\le  \frac{P_{X^\star}(x) P_{Y|X}(y|x) }{P_{X^\star}(x)  P_{Y|X}(y|x) + P_{X^\star}(1-x)  P_{Y|X}(y|1-x) }\\
&=  \frac{P_{Y|X}(y|x) }{  P_{Y|X}(y|x) +  P_{Y|X}(y|1-x) }\\
&= \frac{1}{1 + \left( \frac{x}{1-x} \right)^{n-2y}}.
\end{align}

Next, by recognizing that the function $y\mapsto \log\frac{1}{1+\left( \frac{x}{1-x} \right)^{n-2y}}$ is concave for any $x\in(0,1)$, by applying Jensen's inequality we get:
\begin{align}
 &\expcnd{\log(P_{X^\star|Y}(x|Y))}{X=x} \notag\\
 &\le \bbE \left[ \log \frac{1}{1 + \left( \frac{x}{1-x} \right)^{n-2Y}} |X =x  \right]  \label{eq:Jensen_crest}  \\
 &\le \log \frac{1}{1 + \left( \frac{x}{1-x} \right)^{n-2\expcnd{Y}{X=x}}} \\
 &= \log \frac{1}{1 + \left( \frac{x}{1-x} \right)^{n(1-2x)}}.
\end{align}

This concludes the proof. 
\end{proof}


\begin{figure}
    \centering
    \input{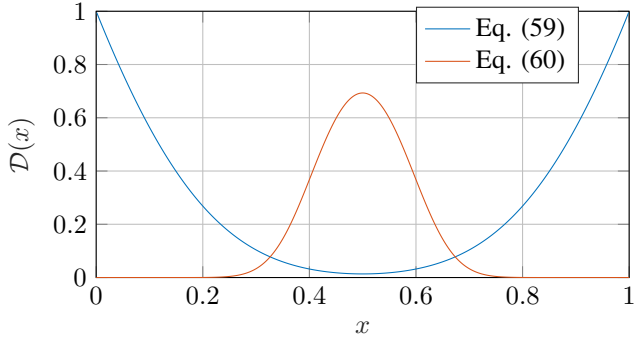}
    \caption{Comparison of the bounds on the crest-factor $\cD(x)$ reported in Proposition~\ref{prop:crest_factor_bounds} for $n=10$.}
    \label{fig:crest_factor_bound}
\end{figure}

In Fig.~\ref{fig:crest_factor_bound}, for the case $n=10$ we show a comparison between the bounds on $\cD(x)$ given in Proposition~\ref{prop:crest_factor_bounds}. The upper bounds on $\cD(x)$ so far have been elusive.

    

\begin{table*}[h!]
    \centering
    \begin{tabular}{|c|c|c|c|c|}
    \Xhline{2pt}
    $n$ & $C(n)$ & ${\cal X}\equiv\supp(P_{X^\star})$   & $\{P_{X^\star}(x),\, x \in {\cal X} \}$ & $\{P_{Y^\star}(y), \, y\in\{0\} \cup [n] \}$ \\
    \Xhline{1pt}
    $1$ & $\log(2)$ & $\left\{0,1 \right\}$  & $\left\{\frac{1}{2},\frac{1}{2}\right\}$     &  $\left\{\frac{1}{2},\frac{1}{2}\right\}$\\
    \hline
    $2$ &  $\log\left(\frac{17}{8}\right)$ & $\left\{0, \frac{1}{2} ,1\right\}$   & $\left\{ \frac{15}{34},  \frac{2}{17},\frac{15}{34}\right\}$  &  $\left\{ \frac{8}{17}, 
 \frac{1}{17},\frac{8}{17} \right\} $ \\
    \hline
    $3$&  $\log\left(\frac{19}{8}\right)$ & $\left\{0, \frac{1}{2} ,1\right\}$ & $\left\{  \frac{15}{38},  \frac{4}{19}, \frac{15}{38}\right\}$ & $\left\{ \frac{8}{19}, 
 \frac{3}{38},\frac{3}{38},\frac{8}{19} \right\}$ \\
   
    \Xhline{1pt}
\end{tabular}
    \caption{Capacity and capacity-achieving distributions.  }
    \label{tab:CapacityExpression}
    
\end{table*}

\subsection{Bounds on the Cardinality} \label{sec:card_bounds}
We now provide upper and lower bounds on the cardinality of the support of $P_{X^\star}$.  We start with the following exact formula for the number of support points. 
\begin{prop}  \label{prop:ExactNumSup}
For $n \ge 1$
    \begin{equation}
| \supp(P_{X^\star})| = \frac{\rme^{ C(n) }}{    \bbE \left[ \rme^{  - \cD(U^*)  } 
 \right]},
\end{equation}
where $U^*$ is uniformly distributed on $\supp(P_{X^\star})$.
\end{prop}
\begin{proof}
   Starting with Lemma~\ref{lem:expression_for_prob} and  summing over $x^{\star} \in \supp(P_{X^{\star}})$, we arrive at
\begin{equation}
1 =   \rme^{ -C(n) } \sum_{ x^{\star} \in \supp(P_{X^{\star}}) }  \rme^{  -  \cD(x^\star) }. \label{eq:Intermideiate_step}
\end{equation}
Dividing both sides of  \eqref{eq:Intermideiate_step} by $| \supp(P_{X^\star})|$ and rearranging,  we arrive at the desired result.  
\end{proof}

From Proposition~\ref{prop:ExactNumSup} and  non-negativity of $\cD$, we arrive at 
\begin{equation}
    | \supp(P_{X^\star})| \ge \rme^{C(n)} =\Theta( \sqrt{n}) \label{eq:sqrt(n) lower bound}
\end{equation}
where the order of the lower bound follows from the fact that $C(n)$ scales as $\frac{1}{2} \log(n)$ as will be shown in Theorem~\ref{thm:bounds_on_cap} (Section~\ref{sec:bounds_on_cap}  below). 

We now  move on to showing upper bounds. We already have demonstrated a bound of order $n+1$ in Proposition \ref{prop:bound_cardinality_cA}.  We now improve this upper bound by a factor of two. 
\begin{theorem} \label{thm:n/2 bound}
    For $n \ge 1$
    \begin{align}
        |\supp{(P_{X^\star})}| &\le  2 + \left\lfloor \frac{1}{2}\sfN\left((0,1); \, i''(x;P_{Y^\star}) \right) \right\rfloor \label{eq:i" bound} \\
        & \le 2+\left\lfloor\frac{n}{2} \right\rfloor. \label{eq: n/2 bound}
    \end{align}
\end{theorem}
\begin{proof}

    First of all, note that by Proposition~\ref{prop:0_is_optimal} we know that the function $i(\cdot; P_{Y^\star})$ starts with a local maximum at $x^\star =0$, and from Proposition~\ref{prop:derivatives_info_density} we know that $i''(x; P_{Y^\star})>0$  for $x\to 0$. Now, by continuity of $i(\cdot; P_{Y^\star})$,  if $i(\cdot; P_{Y^\star})$ changes concavity $k$ times, then it has at most $2+\left\lfloor \frac{k}{2} \right\rfloor$ local maxima.
    Moreover, from the KKT conditions we know that all the zeros of $i(\cdot;P_{Y^\star})-C(n)$ are local maxima. 
    
 Then, we can write 
        \begin{align}
        |\supp{(P_{X^\star})}|  
         & \le \sfN\lr{ \, [0,1]; \, i(\cdot;P_{Y^\star})-C(n) \, } \\
        &\le 2 + \left\lfloor \frac{1}{2}\sfN\left((0,1); \, i''(\cdot;P_{Y^\star}) \right) \right\rfloor \\
        &\le 2+\left\lfloor\frac{n}{2} \right\rfloor, \label{eq:bound_zeros_i''} 
    \end{align}
    where 
    \eqref{eq:bound_zeros_i''} follows from the fact that $x\mapsto x(x-1)i''(x;P_{Y^\star})$ is a polynomial of degree $n$ (see Proposition \ref{prop:derivatives_info_density}). 
\end{proof}

A few remarks are now in order:
\begin{itemize}[leftmargin =*]
    \item The proof of Theorem~\ref{thm:n/2 bound} does not rely on the uniqueness of $P_{X^*}$. Therefore, it improves on the Witsenhausen bound by a factor of two. Furthermore, the key part of the proof leading to \eqref{eq:i" bound} is independent of the fact that the channel is binomial: Indeed, this fact is only used in \eqref{eq: n/2 bound}.   Consequently, we posit that this bound may prove more beneficial for channels where it is feasible to establish bounds on the number of zeros in~$i''(x; P_{X^\star})$.
    \item The lower bound in \eqref{eq:sqrt(n) lower bound} and the upper bound in \eqref{eq: n/2 bound} do not match in their respective orders. This lack of alignment is perhaps unsurprising, considering the inherent difficulty in establishing tight bounds on the cardinality of the support. For further exploration of this challenging problem, the interested reader is directed to \cite{dytso2019capacity,dytso2021properties, kashyap2020many}. We suspect that neither the upper nor the lower bounds are tight. 
\end{itemize}

\section{Capacity and Bounds on the Capacity}
\label{sec:cap_bounds}

In this section, we provide exact  values of the capacity for $n \le 3$. For the remaining regime we provide upper and lower bounds on capacity.

\subsection{Exact Capacity for $n \le 3$}

The exact capacity can be computed by first making a guess of the capacity-achieving distribution according to the properties outlined in Section~\ref{sec:main_I}. Then, this guess can be checked against the sufficient and necessary KKT conditions in Lemma~\ref{lem:KKT}. These, somewhat tedious, computations  are performed in Appendix~\ref{sec:computations_of_cap_exact}  and Table~\ref{tab:CapacityExpression} displays the results.

\subsection{Bounds on the Capacity}
\label{sec:bounds_on_cap}

We now provide  bounds on the capacity. Our upper bound relies on the dual representation of the capacity as:
\begin{equation}
   C(n) = \inf_q \max_{x \in [0,1]} \kl{P_{Y|X}(\cdot|x)}{q},
\end{equation}
which, by properly choosing an auxiliary output distribution \(q\), often leads to order-tight bounds. The  reader is referred to \cite{lapidoth2008capacity,mckellips2004simple,boundsAmplt} for applications to other channels. It will also be convenient to work with continuous output, and we will use the following channel output: $\tilde{Y} = Y + U$, where $U \sim \mathcal{U}(0,1)$. Note that because the distance between original $Y$ points is one, such additive noise can be completely filtered out, and we have $I(X; Y) = I(X; Y+ U)$ for all $X$. This trick has been used before in the context of the Poisson channel in \cite{lapidoth2008capacity}.

The lower bound on the capacity will follow from choosing a convenient input distribution. The exact computation, however, will not be possible, and some further bounds on the entropy of the binomial distribution will be needed. Therefore,  in Appendix~\ref{app:bound_binom_entropy}, we also provide a new upper bound on the entropy of a binomial distribution. Bounds   on the entropy of a binomial distribution have been considered before in \cite{knessl1998integral,cheraghchi2018expressions}.

\begin{theorem} \label{thm:bounds_on_cap} For $n\ge 1$, the channel capacity is bounded below by
    \begin{align} \label{eq:capacity_lower_bound}
        C(n) &\ge \max\left\{\log(2), \log(\pi n)-\frac{1}{2}\log\left(2\pi \rme \left( \frac{n}{8}+\frac{1}{12}\right) \right) \right. \nonumber\\
        &\quad\left.+\frac{1}{\sqrt{\pi \left(n+\frac{1}{4}\right)}}\log\left( \frac{1}{16n^2}  \right)   -\log(4)- 1\right\}
    \end{align}
    and bounded above by
    \begin{align} \label{eq:capacity_upper_bound}
       & C(n) \notag\\
        &\le \min\left\{  
 \log \left( 3+\left\lfloor\frac{(n-1)}{2} \right\rfloor \right) , \, \log(\pi(n+1))-\frac{1}{2}\log( n) \right.\nonumber\\
        &\quad\left.+  \frac{3}{2}+\frac{1}{2^{n+1}}\log\left(n\right) +\frac{1}{2}\log\left( \frac{3}{2 }\left( 1+\frac{1}{n}\right)\right)\right\}. 
    \end{align}
\end{theorem}
\begin{proof}

   A first lower bound follows from observing that $P_{Y|X}(0|0)=1$ and $P_{Y|X}(n|1)=1$. Hence the input distribution $P_X(0)=P_X(1)=\frac{1}{2}$ gives a mutual information of $I(X;Y)=\log(2)$ nats for all $n\ge 1$.

    For an alternative capacity lower bound, pick an input pdf as
\begin{equation}\label{eq:Beta_distribution_f_X}
    f_X(x) = \frac{1}{\pi \sqrt{x(1-x)}}, \qquad x \in (0,1),
\end{equation}
which is a Beta distribution with shape parameters $\alpha = \beta = \frac{1}{2}$. 
 Then, the capacity can be lower-bounded as follows:
\begin{align}
    &C(n) = \max_{P_X} I(X; Y) \\
    &\ge I(X;Y) \\
    &= H(Y)-H(Y|X) \\
    &\ge H(Y)-\frac{1}{2}\expect{\log\left(2\pi \rme \left( nX(1-X)+\frac{1}{12}\right) \right)} \label{eq:use_bound_binomial_rv} \\
    &\ge H(Y)-\frac{1}{2}\log\left(2\pi \rme \left( n\expect{X(1-X)}+\frac{1}{12}\right) \right) \label{eq:apply_Jensen} \\
    &= H(Y)-\frac{1}{2}\log\left(2\pi \rme \left( \frac{n}{8}+\frac{1}{12}\right) \right)
\end{align}
where in~\eqref{eq:use_bound_binomial_rv} we have used the upper bound on the entropy of a binomial distribution, which is given in Appendix~\ref{app:bound_binom_entropy}; in~\eqref{eq:apply_Jensen} we have applied Jensen's inequality; and in the last step we have used that $\expect{X(1-X)}=\frac{1}{8}$ from the Beta distribution in~\eqref{eq:Beta_distribution_f_X}.

As for the output entropy, write
\begin{align}
    H(Y) &= -\expect{\log P_Y(Y)} \\
    &= -\expect{\log\binom{n}{Y}}-\expect{\log \bbE_X[X^Y(1-X)^{n-Y}] }.
\end{align}
Now note that 
\begin{align}
    \expect{X^y(1-X)^{n-y}} &= \int_0^1 \frac{1}{\pi \sqrt{x(1-x)}} x^y(1-x)^{n-y} \rmd x \\
    &= \frac{1}{\pi n} \frac{\Gamma\left(y+\frac{1}{2}\right)\Gamma\left(n-y+\frac{1}{2}\right)}{\Gamma\left(n\right)}
\end{align}
which, by expanding the binomial coefficient in terms of gamma functions, gives
\begin{align}
    &H(Y) = 
     \log(\pi n) \hspace{-0.04cm} - \hspace{-0.04cm} \expect{ \log\left( \hspace{-0.04cm}n \frac{\Gamma\left(Y+\frac{1}{2}\right)}{\Gamma\left(Y+1\right)} \frac{\Gamma\left(n-Y+\frac{1}{2}\right)}{\Gamma\left(n-Y+1\right)} \right) } \\
    &\ge \log(\pi n) -\expect{ \log\left( \frac{n}{\left(Y+\frac{1}{4}\right)^{\frac{1}{2}} \left(n-Y+\frac{1}{4}\right)^{\frac{1}{2}}}  \right) } \label{eq:use_Kershav}\\
    &= \log(\pi n) +\frac{1}{2}\expect{ \log\left( \left(\frac{Y}{n}+\frac{1}{4n}\right)\left(1-\frac{Y}{n}+\frac{1}{4n} \right)  \right) } \\
    &= \log(\pi n) +\expect{ \log\left( \frac{Y}{n}+\frac{1}{4n}  \right) } \label{eq:use_symmetry} \\
    &\ge \log(\pi n)+\expect{\mathbbm{1}(Y=0) \log\left( \frac{1}{4n}  \right) } \nonumber\\
    &\quad+\expect{\mathbbm{1}(0<Y\le n) \log\left( \frac{Y}{n}  \right) } \\
    &\ge \log(\pi n)+\expect{(1-X)^n}\log\left( \frac{1}{4n}  \right) \nonumber\\
    &\quad+\expect{(1-(1-X)^n)\log(X)} -1 \label{eq:use_lemma_exp_log_bino}\\
    &\ge \log(\pi n)+\frac{\Gamma\left(n+\frac{1}{2}\right)}{\sqrt{\pi} \Gamma\left(n+1\right)}\log\left( \frac{1}{16n^2}  \right)   +\expect{\log(X)} -1 \label{eq:nth_moment_beta}  \\
    &\ge \log(\pi n)+\frac{1}{\sqrt{\pi \left(n+\frac{1}{4}\right)}}\log\left( \frac{1}{16n^2}  \right)   -\log(4)- 1\label{eq:use_Kershav2}
\end{align}
where in~\eqref{eq:use_Kershav} and in~\eqref{eq:use_Kershav2} we used Kershav's inequality~\cite{kershaw1983some}
\begin{equation}
    \frac{\Gamma(x+s)}{\Gamma(x+1)} \le \frac{1}{\left(x+\frac{s}{2}\right)^{1-s}}
\end{equation}
for $x>0$ and $s \in (0,1)$; in~\eqref{eq:use_symmetry} we have used the symmetry of the output pmf $Y \stackrel{d}{=} (n-Y)$; in~\eqref{eq:use_lemma_exp_log_bino} we have used Lemma~\ref{lem:expect_log_binomial}; in~\eqref{eq:nth_moment_beta} we have  have used $\expect{(1-X)^n}=\frac{\Gamma\left(n+\frac{1}{2}\right)}{\sqrt{\pi} \Gamma\left(n+1\right)}$ and the fact that $\log(X)\le 0$; finally, in the last step we have used $\expect{\log(X)}=-\log(4)$.

To sum up, the capacity lower bound is given by \eqref{eq:capacity_lower_bound}.

A first upper bound on $C(n)$ follows by noting that 
\begin{equation}
C(n) \le H(X^\star) \le \log \left( 3+\left\lfloor\frac{(n-1)}{2} \right\rfloor \right),
\end{equation}
where the last upper bound is due to Theorem~\ref{thm:n/2 bound}. 

For an alternative capacity upper bound, choose the auxiliary output pdf
\begin{equation}
    q(t) = \frac{1}{\pi(n+1)} \left(\frac{t}{n+1}\left(1-\frac{t}{n+1}\right)\right)^{-\frac{1}{2}}, \, t \in (0,n+1), \label{eq:choice_auxiliary_q}
\end{equation}
and, by introducing $U \sim {\cal U}[0,1]$ independent of $X$ and $Y$, write
\begin{align}
    &C(n) = \max_{P_X} I(X;Y) \\
    &= \max_{P_X} I(X; Y+U) \\
    &\le \max_{x \in [0,1]} \kl{P_{Y+U|X}(\cdot|x)}{q} \label{eq:use_duality} \\
    &= \max_{x \in [0,1]} -H(Y|X=x) -\expcnd{\log q(Y+U)}{X=x} \label{eq:continuous_and_discrete_entropy} \\
    &\le \max_{x \in [0,1]} \left\{ \log(\pi(n+1))-(1-(1-x)^n-x^n)\frac{1}{2}\log\left( 2\pi n\right) \right.\nonumber\\
    &\quad-\frac{1}{2}(1-(1-x)^n)\log(x)-\frac{1}{2}(1-x^n)\log(1-x)+1  \nonumber\\
    &\quad \left.+\frac{1}{2}\expcnd{\log\left( \frac{Y+U}{n+1} \left( 1-\frac{Y+U}{n+1}\right)\right)}{X=x} \right\} \label{eq:use_bound_on_binomial_entropy} \\
    &\le \max_{x \in [0,1]} \left\{\log(\pi(n+1)) -(1-(1-x)^n-x^n)\frac{1}{2}\log\left( 2\pi n\right) \right. \nonumber\\
    &\quad -\frac{1}{2}(1-(1-x)^n)\log(x)-\frac{1}{2}(1-x^n)\log(1-x)+1  \nonumber\\
    &\quad \left.+\frac{1}{2}\log\left( \frac{nx+\frac{1}{2}}{n+1} \left( 1-\frac{nx+\frac{1}{2}}{n+1}\right)\right) \right\} \label{eq:capacity_upper_bound_1} \\
    &=\log(\pi(n+1))-\frac{1}{2}\log(2\pi n) +\max_{x \in [0,1]} \left\{ g_n(x)\right\} \label{eq:introduce_gn}
\end{align}
where in~\eqref{eq:use_duality} we have used the dual formulation of capacity; in~\eqref{eq:continuous_and_discrete_entropy} we used $h(Y+U|X=x)=H(Y|X=x)$; in~\eqref{eq:use_bound_on_binomial_entropy} we have used the lower bound on the entropy of a binomial distribution given in Appendix
\ref{app:bound_binom_entropy}; in \eqref{eq:capacity_upper_bound_1} we have used Jensen's inequality; and in  \eqref{eq:introduce_gn} we have introduced the function
\begin{align}
    &g_n(x) = \frac{((1-x)^n+x^n)}{2}\log\left( 2\pi n\right)-\frac{(1-(1-x)^n)}{2}\log(x) \nonumber\\
    &\, \, \,-\frac{(1-x^n)}{2}\log(1-x)+\frac{1}{2}\log\left( \frac{nx+\frac{1}{2}}{n+1} \left( 1-\frac{nx+\frac{1}{2}}{n+1}\right)\right)
\end{align}
for $x \in [0,1]$. In Appendix \ref{app_bound_on_g}, we prove a uniform upper bound on $g_n$. As a result, the capacity upper bound is given in \eqref{eq:capacity_upper_bound}.
\end{proof}

\bibliographystyle{IEEEtran}
\bibliography{refs.bib}

\begin{appendices}

\section{Derivatives of Information Density}
\label{sec:Derivatives_info_density}

\begin{lemma}\label{lem:derivative_exp_cnd}
For any function $f:\bbN_0 \cap [0,n]\to \bbR$, and $x \in (0,1)$, we have
\begin{align}
   & \frac{\rmd}{\rmd x}\expcnd{f(Y)}{X=x} \notag\\
    &= n \expcndn{n-1}{f(Y+1)-f(Y)}{X=x} \label{eq:der_cond_exp_n-1}\\
    &=
    \frac{1}{x}\expcnd{Yf(Y)}{X=x}-\frac{1}{1-x}\expcnd{(n-Y)f(Y)}{X=x} \label{eq:der_cond_exp_0}\\
    &= \frac{1}{x}\expcnd{Y(f(Y)-f(Y-1))}{X=x} \label{eq:der_cond_exp_1} \\
    &= \frac{1}{1-x}\expcnd{(n-Y)(f(Y+1)-f(Y))}{X=x},\label{eq:der_cond_exp_2}
\end{align}
\end{lemma}
\begin{proof}
Let us first prove \eqref{eq:der_cond_exp_n-1}. Note that 
\begin{align}
    \frac{\rmd}{\rmd x} P_{Y|X}(y|x) &= {n \choose y} \frac{\rmd}{\rmd x} x^y (1-x)^{n-y} \\
    &=\frac{y-nx}{x(1-x)}{n \choose y}  x^y (1-x)^{n-y}  \\
    &= \frac{y-nx}{x(1-x)}P_{Y|X}(y|x) \\
    &= \left(\frac{y}{x}-\frac{n-y}{1-x}\right)P_{Y|X}(y|x).
\end{align}
Hence, we can write
\begin{align}
    &\frac{\rmd}{\rmd x}\expcnd{f(Y)}{X=x} \notag\\
    &= \sum_{y=0}^n f(y) \frac{\rmd}{\rmd x} P_{Y|X}(y|x) \\
    &= \sum_{y=0}^n f(y) \left(\frac{y}{x}-\frac{n-y}{1-x}\right)P_{Y|X}(y|x) \label{eq:obtain_der_cond_exp_0} \\
    &= \sum_{y=0}^{n-1} f(y+1) \frac{y+1}{x}P_{Y|X}(y+1|x) \notag\\
    & \quad -\sum_{y=0}^{n-1} f(y) \frac{n-y}{1-x}P_{Y|X}(y|x)\\
    &= \sum_{y=0}^{n-1} f(y+1) nP_{Y|X}^{n-1}(y|x)-\sum_{y=0}^{n-1} f(y) nP_{Y|X}^{n-1}(y|x)\\
    &= n \expcndn{n-1}{f(Y+1)-f(Y)}{X=x}.
\end{align}
Result~\eqref{eq:der_cond_exp_0} is obtained in~\eqref{eq:obtain_der_cond_exp_0}. 

To prove \eqref{eq:der_cond_exp_1}, first note that
\begin{align}
    & \frac{\rmd}{\rmd x}\expcnd{f(Y)}{X=x} \notag\\
    &=   \sum_{y=0}^n f(y) {n \choose y}\frac{\rmd}{\rmd x} x^y (1-x)^{n-y} \\
    &= \sum_{y=0}^n f(y) {n \choose y} x^y (1-x)^{n-y} \left(\frac{y}{x}-\frac{n-y}{1-x} \right). \label{eq:diff_expcnd}
\end{align}
Next, consider the term
\begin{align}
    &\sum_{y=0}^n f(y) {n \choose y} x^y (1-x)^{n-y} \frac{n-y}{1-x}  \notag\\
    &= \sum_{y=0}^{n-1} f(y) {n \choose y} x^{y+1} (1-x)^{n-(y+1)} \frac{n-y}{x} \\
    &= \sum_{y=0}^{n-1} f(y) {n \choose y+1} x^{y+1} (1-x)^{n-(y+1)} \frac{y+1}{x} \label{eq:binomchange1} \\
    &= \sum_{y=1}^{n} f(y-1) {n \choose y} x^{y} (1-x)^{n-y} \frac{y}{x} \\
    &= \frac{1}{x}\expcnd{Yf(Y-1)}{X=x} \label{eq:first_expcnd_rel}
\end{align}
where \eqref{eq:binomchange1} follows from ${n \choose y} = {n \choose y+1}\frac{y+1}{n-y}$ for $y \ne n$. Using \eqref{eq:first_expcnd_rel} into \eqref{eq:diff_expcnd} proves \eqref{eq:der_cond_exp_1}.

Now consider the term
\begin{align}
    &\sum_{y=0}^n f(y) {n \choose y} x^y (1-x)^{n-y} \frac{y}{x}  \notag\\
    &= \sum_{y=1}^{n} f(y) {n \choose y} x^{y-1} (1-x)^{n-(y-1)} \frac{y}{1-x} \\
    &= \sum_{y=1}^{n} f(y) {n \choose y-1} x^{y-1} (1-x)^{n-(y-1)} \frac{n-y+1}{1-x} \label{eq:binomchange2} \\
    &= \sum_{y=0}^{n-1} f(y+1) {n \choose y} x^{y} (1-x)^{n-y} \frac{n-y}{1-x} \\
    &= \frac{1}{1-x}\expcnd{(n-Y)f(Y+1)}{X=x} \label{eq:second_expcnd_rel}
\end{align}
where \eqref{eq:binomchange2} follows from ${n \choose y} = {n \choose y-1}\frac{n-y+1}{y}$ for $y \ne 0$. Using \eqref{eq:second_expcnd_rel} into \eqref{eq:diff_expcnd} proves \eqref{eq:der_cond_exp_2}.
\end{proof}

\begin{lemma}\label{lem:channel_transformations}
    For all $y=0,1,\ldots, n-1$, we have
    \begin{align}
        P_{Y|X}(y+1|x) &= \frac{nx}{y+1}P_{Y|X}^{n-1}(y|x) \label{eq:tr_ch_1} \\
        P_{Y|X}(y|x) &= \frac{n(1-x)}{n-y}P_{Y|X}^{n-1}(y|x). \label{eq:tr_ch_2}
    \end{align}
    Moreover, we have
    \begin{align}
        P_Y(y+1) &= \frac{n}{y+1} \expcndn{n-1}{X}{Y=y} P_Y^{n-1}(y), \label{eq:tr_P_Y_1} \\
        P_Y(y) &= \frac{n}{n-y} \expcndn{n-1}{1-X}{Y=y} P_Y^{n-1}(y), \label{eq:tr_P_Y_2} \\
        \frac{n-y}{y+1}\frac{P_Y(y)}{P_Y(y+1)} &= \frac{\expcndn{n-1}{1-X}{Y=y}}{\expcndn{n-1}{X}{Y=y}} \\
        &= \frac{\expcnd{1-X}{Y=y+1}}{\expcnd{X}{Y=y}}. \label{eq:ratio_cond_exp_n}
    \end{align}
\end{lemma}
\begin{proof}
    To show~\eqref{eq:tr_ch_1}, write
    \begin{align}
        P_{Y|X}(y+1|x) &= {n \choose y+1} x^{y+1} (1-x)^{n-y-1} \\
        &= \frac{nx}{y+1}{n-1 \choose y} x^{y} (1-x)^{n-y-1} \\
        &= \frac{nx}{y+1} P_{Y|X}^{n-1}(y|x).
    \end{align}
    To show~\eqref{eq:tr_P_Y_1}, write
    \begin{align}
        P_Y(y+1) &= \sum_x P_X(x) P_{Y|X}(y+1|x) \\
        &= \sum_x P_X(x)\frac{nx}{y+1}P_{Y|X}^{n-1}(y|x) \label{eq:apply_channel_downgrage} \\
        &= \frac{n}{y+1} \expcndn{n-1}{X}{Y=y} P_Y^{n-1}(y) \label{eq:multiply_and_divide_by_p_Y}
    \end{align}
    where in~\eqref{eq:apply_channel_downgrage} we used~\eqref{eq:tr_ch_1}; and in the last step we multiplied and divided by $P_Y^{n-1}(y)$.
    In a similar fashion, to show~\eqref{eq:tr_ch_2}, write
    \begin{align}
        P_{Y|X}(y|x) &= {n \choose n-y} x^{y} (1-x)^{n-y} \\
        &= \frac{n(1-x)}{n-y}{n-1 \choose n-y-1} x^{y} (1-x)^{n-y-1} \\
        &= \frac{n(1-x)}{n-y}P_{Y|X}^{n-1}(y|x).
    \end{align}
    To show~\eqref{eq:tr_P_Y_2}, write
    \begin{align}
        P_Y(y) &= \sum_x P_X(x) P_{Y|X}(y|x) \\
        &= \sum_x P_X(x)\frac{n(1-x)}{n-y}P_{Y|X}^{n-1}(y|x) \label{eq:apply_channel_upgrade} \\
        &= \frac{n}{n-y} \expcndn{n-1}{1-X}{Y=y} P_Y^{n-1}(y) \label{eq:multiply_and_divide_by_p_Y_2}
    \end{align}
    where in~\eqref{eq:apply_channel_upgrade} we used~\eqref{eq:tr_ch_2}; and in the last step we multiplied and divided by $P_Y^{n-1}(y)$. 
    Finally, to show~\eqref{eq:ratio_cond_exp_n},  write
    \begin{align}
       & P_Y(y)\expcnd{X}{Y=y} \notag\\
        &= \sum_x P_X(x) x  {n \choose y} x^{y} (1-x)^{n-y} \\
        &= \sum_x P_X(x) (1-x)  {n \choose y} x^{y+1} (1-x)^{n-(y+1)} \\
        &= \sum_x P_X(x) (1-x)\frac{y+1}{n-y}  {n \choose y+1} x^{y+1} (1-x)^{n-(y+1)} \\
        &=\frac{y+1}{n-y} P_Y(y+1) \left(1-\expcnd{X}{Y=y+1} \right)
    \end{align}
    which is valid for $y\ne n$.
\end{proof}
\begin{prop}\label{prop:der_info_den}
For $n\ge 1$ and $x \in (0,1)$, we have
\begin{align}
    &i'(x;P_Y)  =n\log\left(\frac{x}{1-x}\right) \notag\\
    &+ \frac{1}{1-x} \expcnd{(n-Y)\log\frac{\expcnd{1-X}{Y=Y+1}}{\expcnd{X}{Y=Y}}}{X=x}.
\end{align}
For $n\ge 2$ and $x \in (0,1)$, we have
\begin{align}
   & i'(x;P_Y)  \notag\\
    &= n\log\left(\frac{x}{1-x}\right)+ \hspace{-0.1cm}n \expcndn{n-1}{\log\frac{\expcndn{n-1}{1-X}{Y}}{\expcndn{n-1}{X}{Y}}}{X=x} \label{eq:inf_den_derivative_no_bregman} \\
    &= \frac{n \, \expcndn{n-1}{\ell_b(x,\expcndn{n-1}{X}{Y})+\frac{x-\expcndn{n-1}{X}{Y}}{1-\expcndn{n-1}{X}{Y}}}{X=x}  }{x} \label{eq:inf_den_derivative_with_bregman}
\end{align}
and $i''(x;P_Y)$ is given in \eqref{eq:i_doubleprime}. 
\begin{figure*}
\begin{align}
    i''(x;P_Y) 
    &= \frac{n}{x(1-x)}+\frac{\expcnd{(n-Y)(n-Y-1)\log\frac{\expcnd{X}{Y=Y}}{\expcnd{1-X}{Y=Y+1}}\frac{\expcnd{1-X}{Y=Y+2}}{\expcnd{X}{Y=Y+1}}}{X=x} }{(1-x)^2}  \label{eq:i_doubleprime} \\
    &= \frac{n \left(1+\expcndn{n-1}{Y\log\frac{\expcndn{n-1}{1-X}{Y} }{\expcndn{n-1}{X}{Y}}}{X=x} \right)-\frac{n-1}{1-x} \left(i'(x;P_Y)-n\log\left(\frac{x}{1-x}\right) \right) }{x(1-x)}. \label{eq:i''_versions}
\end{align}
\end{figure*}
For $n\ge 3$ and $x \in (0,1)$, we have $i''(x;P_Y)$ is given in \eqref{eq:i''_versions_n>3}. 
\begin{figure*}
\begin{align}
    i''(x;P_Y) &= \frac{n}{x(1-x)}+n(n-1) \expcndn{n-2}{\log\frac{\expcndn{n-1}{1-X}{Y+1}}{\expcndn{n-1}{X}{Y+1}}\frac{\expcndn{n-1}{X}{Y}}{\expcndn{n-1}{1-X}{Y}}}{X=x} \\
    &=\frac{n}{x(1-x)}\left(1+(n-1)\expcndn{n-2}{ \ell_b(x, \expcndn{n-1}{X}{Y+1}) + \frac{x-\expcndn{n-1}{X}{Y+1}}{1-\expcndn{n-1}{X}{Y+1}} }{X=x} \right)-\frac{n-1}{1-x} i'(x;P_Y). \label{eq:i''_versions_n>3}
\end{align}
\end{figure*}
\end{prop}
\begin{proof}
Let us start from the expression
    \begin{align}        i(x;P_{Y})&=\expcnd{\log\frac{{n \choose Y}}{P_Y(Y)} }{X=x}+nx\log(x) \notag\\
    &+ n(1-x)\log(1-x). 
    \end{align}
    A way to compute the first derivative is as follows:
    \begin{align}
        &i'(x;P_{Y}) \notag\\
        &=\frac{\rmd}{\rmd x}\expcnd{\log\frac{{n \choose Y}}{P_Y(Y)}  }{X=x}+n\log\left(\frac{x}{1-x}\right) \\
        &= \frac{\expcnd{(n-Y)\log\frac{\binom{n}{Y+1}}{P_Y(Y+1)}\frac{P_Y(Y)}{\binom{n}{Y}}}{X=x}}{1-x}  \notag\\
        & \quad +n\log\left(\frac{x}{1-x}\right) \label{eq:use_last_identity_der_lemma} \\
        &= \frac{\expcnd{(n-Y)\log\frac{\expcnd{1-X}{Y=Y+1}}{\expcnd{X}{Y=Y}}}{X=x}}{1-x} \notag\\
       & \quad +n\log\left(\frac{x}{1-x}\right), \label{eq:i_prime_first_ver}
    \end{align}
    where~\eqref{eq:use_last_identity_der_lemma} follows from identity~\eqref{eq:der_cond_exp_2} of Lemma~\ref{lem:derivative_exp_cnd}; and the last step follows from identity~\eqref{eq:ratio_cond_exp_n} of Lemma~\ref{lem:channel_transformations}.
    
    An alternative expression can be derived as follows:
    \begin{align}
        &i'(x;P_{Y}) \notag\\
        &=\frac{\rmd}{\rmd x}\expcnd{\log\frac{{n \choose Y}}{P_Y(Y)}  }{X=x}+n\log\left(\frac{x}{1-x}\right) \\
        &= n\expcndn{n-1}{\log\frac{{n \choose Y+1}P_Y(Y)}{{n \choose Y}P_Y(Y+1)}  }{X=x}+n\log\left(\frac{x}{1-x}\right) \label{eq:apply_derivative_downgrade} \\
        &= n\expcndn{n-1}{\log\frac{\expcndn{n-1}{1-X}{Y} }{\expcndn{n-1}{X}{Y}}  }{X=x}+n\log\left(\frac{x}{1-x}\right)
    \end{align}
    where~\eqref{eq:apply_derivative_downgrade} follows from Lemma~\ref{lem:derivative_exp_cnd}, and the last step follows from Lemma~\ref{lem:channel_transformations}.
    To show~\eqref{eq:inf_den_derivative_with_bregman}, just notice that the Bregman divergence for the binomial channel is
    \begin{align} \label{eq:Bregman_divergence_Binomial}
        &\ell_b(x, \expcndn{n-1}{X}{Y}) \notag\\
        &= x \log\left(\frac{x \expcndn{n-1}{1-X}{Y}}{(1-x)\expcndn{n-1}{X}{Y}}\right)-\frac{x-\expcndn{n-1}{X}{Y}}{1-\expcndn{n-1}{X}{Y}}.
    \end{align}
    The computation of the second derivative is performed in \eqref{eq:sec_der_comp_proof}, where in~\eqref{eq:use_iprime_first_ver} we used~\eqref{eq:i_prime_first_ver} and identity~\eqref{eq:der_cond_exp_2} of Lemma~\ref{lem:derivative_exp_cnd}. 
    \begin{figure*}
    \begin{align}
    i''(x;P_Y) &= \frac{n}{x(1-x)}+\frac{1}{(1-x)^2}  \expcnd{(n-Y)\log\frac{\expcnd{1-X}{Y=Y+1}}{\expcnd{X}{Y=Y}}}{X=x} \nonumber\\
    &\quad \frac{1}{(1-x)^2}  \expcnd{(n-Y)\left((n-Y-1)\log\frac{\expcnd{1-X}{Y=Y+2}}{\expcnd{X}{Y=Y+1}}-(n-Y)\log\frac{\expcnd{1-X}{Y=Y+1}}{\expcnd{X}{Y=Y}}\right)}{X=x} \label{eq:use_iprime_first_ver} \\
    &=\frac{n}{x(1-x)}+\frac{1}{(1-x)^2}  \expcnd{(n-Y)(n-Y-1)\log\frac{\expcnd{X}{Y=Y}}{\expcnd{1-X}{Y=Y+1}}\frac{\expcnd{1-X}{Y=Y+2}}{\expcnd{X}{Y=Y+1}}}{X=x}, \label{eq:sec_der_comp_proof}
    \end{align}
    \end{figure*}

    An alternative formulation of the second derivative is as follows:
    \begin{align}
    &i''(x;P_Y) \notag\\
    &= \frac{n}{x(1-x)}+\frac{\rmd}{\rmd x}n\expcndn{n-1}{\log\frac{\expcndn{n-1}{1-X}{Y} }{\expcndn{n-1}{X}{Y}}  }{X=x}\\
    &=\frac{n}{x(1-x)}+n(n-1)\bbE^{n-2} \left[\log\frac{\expcndn{n-1}{1-X}{Y+1} }{\expcndn{n-1}{X}{Y+1}} \right. \\
    & \left. \quad -\log\frac{\expcndn{n-1}{1-X}{Y} }{\expcndn{n-1}{X}{Y}} \mid X =x  \right]
\end{align}
where we applied identity~\eqref{eq:der_cond_exp_n-1} of Lemma~\ref{lem:derivative_exp_cnd}.
Another alternative expression for the second derivative, that is written in terms of the first derivative, is as follows
\begin{align}
    &i''(x;P_Y) \notag\\
    &=\frac{n}{x(1-x)}+ \frac{\rmd}{\rmd x}n\expcndn{n-1}{\log\frac{\expcndn{n-1}{1-X}{Y} }{\expcndn{n-1}{X}{Y}}  }{X=x}\\
    &=\frac{n}{x(1-x)} \notag\\
    &+ n\left(\frac{1}{x}+\frac{1}{1-x}\right)\expcndn{n-1}{Y\log\frac{\expcndn{n-1}{1-X}{Y} }{\expcndn{n-1}{X}{Y}}  }{X=x}\nonumber \\
    &\quad-\frac{n-1}{1-x}n\expcndn{n-1}{\log\frac{\expcndn{n-1}{1-X}{Y} }{\expcndn{n-1}{X}{Y}}  }{X=x} \label{eq:apply_der_exp_cond_2}\\
    &= \frac{n \left(1+\expcndn{n-1}{Y\log\frac{\expcndn{n-1}{1-X}{Y} }{\expcndn{n-1}{X}{Y}}}{X=x} \right)}{x(1-x)}  \notag\\
    & \quad -\frac{n-1}{1-x} \left(i'(x;P_Y)-n\log\left(\frac{x}{1-x}\right) \right) \label{eq:use_first_derivative} \\
    &= \frac{n \left(1+(n-1)x\expcndn{n-2}{\log\frac{\expcndn{n-1}{1-X}{Y+1} }{\expcndn{n-1}{X}{Y+1}}}{X=x} \right)}{x(1-x)} \notag\\
    &\quad -\frac{n-1}{1-x} \left(i'(x;P_Y)-n\log\left(\frac{x}{1-x}\right) \right) \label{eq:use_channel_downgrade} \\
    &= \frac{n \left(1+(n-1)x\expcndn{n-2}{\log\frac{x\expcndn{n-1}{1-X}{Y+1} }{(1-x)\expcndn{n-1}{X}{Y+1}}}{X=x} \right)}{x(1-x)} \notag\\
    &\quad -\frac{n-1}{1-x} i'(x;P_Y)  \\
    &= \frac{n}{x(1-x)} \Bigg(1+(n-1)\bbE^{n-2} \Big [ \ell_b(x, \expcndn{n-1}{X}{Y+1}) \notag\\
    & \quad + \left. \frac{x-\expcndn{n-1}{X}{Y+1}}{1-\expcndn{n-1}{X}{Y+1}} \mid {X=x} \right] \Bigg)-\frac{n-1}{1-x} i'(x;P_Y)
\end{align}
where~\eqref{eq:apply_der_exp_cond_2} follows from identity~\eqref{eq:der_cond_exp_0} of Lemma~\ref{lem:derivative_exp_cnd}; in~\eqref{eq:use_first_derivative} we used result~\eqref{eq:inf_den_derivative_no_bregman}; in~\eqref{eq:use_channel_downgrade} we made a change of measure by using identity~\eqref{eq:tr_ch_1} of Lemma~\ref{lem:channel_transformations}; and in the last step we used the Bregman divergence for the Binomial channel~\eqref{eq:Bregman_divergence_Binomial}.

\end{proof}

\section{Capacity Computation for $n \le 3$}
\label{sec:computations_of_cap_exact}

\subsection{The Case of $n=1$}
Follows immediately from Proposition~\ref{prop:loc_info}.

\subsection{The Case of $n=2$}
From Proposition~\ref{prop:loc_info}, for $n=2$, we infer that $$\supp(P_{X^\star}) \subseteq  \left\{0, \frac{1}{2},1 \right\}.$$

Now let $p= P_{X^\star}(\frac{1}{2})$.  Using Corollary~\ref{cor:rel_cap_probY} and the equations for $P_{Y^\star}$, we have that 
\begin{equation}
    P_{Y^\star}(0) = P_{Y^\star}(2) = \rme^{-C(2)}
\end{equation}
\begin{equation}
    P_{Y^\star}(1) = p 2\frac{1}{2}\left(1-\frac{1}{2}\right) = \frac{p}{2}.
\end{equation}
From $\sum_{y=0}^2 P_{Y^\star}(y)=1$ it follows that:
\begin{equation}
    2 \rme^{-C(2)}+\frac{p}{2} = 1
\end{equation}
or
\begin{equation}
    p = 2(1-2\rme^{-C(2)}).
\end{equation}
From the KKT equality condition in  \eqref{eq:KKT_equality}, we have that
\begin{align}
    C(2) &= i \left(\frac{1}{2};P_{Y^\star} \right) \\
    &= \sum_{y=0}^2 \binom{2}{y} \frac{1}{2^y} \left(1-\frac{1}{2} \right)^{2-y} \log\frac{\binom{2}{y} \frac{1}{2^y} \left(1-\frac{1}{2} \right)^{2-y}}{P_{Y^\star}(y)} \\
    &= \frac{1}{4} \left(\log\frac{1}{4}+C(2) \right) + \frac{1}{2} \log\frac{1}{2(1-2\rme^{-C(2)})}  \notag\\
    & \quad + \frac{1}{4} \left(\log\frac{1}{4}+C(2) \right)
\end{align}
that can be rewritten as
\begin{equation}
    C(2) = \log\frac{1}{4}+\log\frac{1}{2(1-2 \rme^{-C(2)})}
\end{equation}
whose solution is $C(2) = \log\frac{17}{8}$. We also have $p = \frac{2}{17}$ and $P_{X^\star}(0) = P_{X^\star}(1) = \frac{15}{34}$.

\subsection{The Case of $n=3$}

From Proposition~\ref{prop:loc_info} and Proposition~\ref{thm:n/2 bound}, for $n=3$, we infer that $$\supp(P_{X^\star}) \subseteq  \left\{0, \frac{1}{2},1 \right\}.$$

Now let $p= P_{X^\star}(\frac{1}{2})$. Corollary~\ref{cor:rel_cap_probY}  and direction computations imply that   
\begin{align}
P_{Y^\star}(0) &=  P_{Y^\star}(3) =\rme^{-C(3)} ,\\
P_{Y^\star}(1) &=  P_{Y^\star}(2) = \frac{3}{8} p. 
\end{align}
Now using above and the fact that $\sum_{y=0}^3 P_{Y^*}(y)=1$, we have that 
\begin{equation}
    p = \frac{4}{3} \left(1 -2 \rme^{-C(3)} \right).  \label{eq:expression_for_p}
\end{equation}

Next, it can be shown that 
\begin{align}
    i \left(\frac{1}{2};P_{Y^\star} \right)  = \frac{1}{4} \log \left(  \frac{\rme^{C(3)}}{8p^3} \right)  .
\end{align}

From the KKT equality condition in  \eqref{eq:KKT_equality}, we have that 
\begin{equation}
  C(3)=  i \left(\frac{1}{2};P_{Y^\star} \right)  = \frac{1}{4} \log \left(  \frac{\rme^{C(3)}}{8p^3} \right)  
\end{equation}
using the expression for $p$ in \eqref{eq:expression_for_p} and simplifying, we arrive at
\begin{equation}
    C(3) = \log \left(  \frac{1}{ \frac{8}{3} \left( 1- 2\rme^{-C(3)} \right)}\right). 
\end{equation}
Solving for $C(3)$ we arrive at
\begin{equation}
    C(3) =\log \left( \frac{19}{8} \right) . 
\end{equation}
We also have that 
\begin{align}
P_{Y^\star}(0) &=  P_{Y^\star}(3) =\frac{8}{19} ,\\
P_{Y^\star}(1) &=  P_{Y^\star}(2) = \frac{3}{38}, \\
P_{X^*}(0) &= P_{X^*}(1) =\frac{15}{38},\\
P_{X^*} \left( \frac{1}{2} \right) &= \frac{4}{19}.
\end{align}

\section{A Uniform Bound on $g_n(x)$} \label{app_bound_on_g}
\begin{lemma}
    For $n\ge 1$, we have
    \begin{align} \label{eq:uniform_bound_on_g}
        \max_{x \in [0,1]} g_n(x) &\le \frac{1}{2}\log(2\pi) +\frac{1}{2}+\frac{1}{2^{n+1}}\log\left(n\right) \notag\\
        &\quad +\frac{1}{2}\log\left(\frac{3}{2} \left( 1+\frac{1}{n}\right) \right).
    \end{align}
\end{lemma}
\begin{proof}
    Since $g_n(x) = g_n(1-x)$, we can limit our analysis in the interval $x\in \left[0,\frac{1}{2}\right]$. 
By applying the substitution $x = \frac{\alpha}{n}$ for $\alpha \in \left[0,\frac{n}{2}\right]$, we get
\begin{align}
    &2g_n\left(\frac{\alpha}{n}\right) \\
    &=\left(\left(1-\frac{\alpha}{n}\right)^n+\left(\frac{\alpha}{n}\right)^n\right)\log\left( 2\pi n\right) \notag\\
    &-\left(1-\left(1-\frac{\alpha}{n}\right)^n\right)\log\left(\frac{\alpha}{n}\right)-\left(1-\left(\frac{\alpha}{n}\right)^n\right)\log\left(1-\frac{\alpha}{n}\right) \nonumber\\
    &\quad+\log\left( \frac{\alpha+\frac{1}{2}}{n+1} \left( 1-\frac{\alpha+\frac{1}{2}}{n+1}\right)\right) \\
    &\le \log\left(\frac{2\pi n}{n+1}\right)-\left(1-\left(1-\frac{\alpha}{n}\right)^n\right)\log\left(\alpha\right) \notag\\
    & \quad -\left(1-\left(\frac{\alpha}{n}\right)^n\right)\log\left(n-\alpha\right)  \notag\\
    &\quad +\log\left( \alpha+\frac{1}{2}\right)+\log \left( n-\alpha+\frac{1}{2}\right) \\
    &\le \log(2\pi) -\left(1-\left(1-\frac{\alpha}{n}\right)^n\right)\log\left(\alpha\right) \notag\\
    &\quad +\left(\frac{\alpha}{n}\right)^n\log\left(n-\alpha\right) +\log\left( \alpha+\frac{1}{2}\right) \notag\\
    &\quad +\log \left( 1+\frac{1}{2(n-\alpha)}\right) \\
    &\le \log(2\pi) -\left(1-\left(1-\frac{\alpha}{n}\right)^n\right)\log\left(\alpha\right)+\frac{1}{2^n}\log\left(n\right) 
 \notag\\
 &\quad +\log\left( \alpha+\frac{1}{2}\right)+\log \left( 1+\frac{1}{n}\right).
\end{align}
When $\alpha \in [0,1]$, the term $\log(\alpha)$ is negative, and by using $\left(1-\frac{\alpha}{n}\right)^n\ge 1-\alpha$ we can further upper-bound as follows:
\begin{align}
    2g_n\left(\frac{\alpha}{n}\right) &\le \log(2\pi) -\alpha \log(\alpha)+\frac{1}{2^n}\log\left(n\right) +\log\left( \alpha+\frac{1}{2}\right) \notag\\
    &\quad +\log \left( 1+\frac{1}{n}\right) \\
    &\le \log(2\pi)+ \rme^{-1}+\frac{1}{2^n}\log\left(n\right) +\log\left( 1+\frac{1}{2}\right) \notag\\
    \quad &+\log \left( 1+\frac{1}{n}\right), \label{eq:upper_on_g_n_looser}
\end{align}
which is bounded in $n$.  When $\alpha \in \left[1,\frac{n}{2}\right]$, then $\log(\alpha)$ is positive and, by using $e^{-x} \ge \left(1-\frac{x}{n} \right)^n$ for all $n \ge 1$ and all $x \ge 0$, we have that 
\begin{align}
 2g_n\left(\frac{\alpha}{n}\right)&\le  \log(2\pi) +\left(\left(1-\frac{\alpha}{n}\right)^n-1\right)\log\left(\alpha\right)+\frac{1}{2^n}\log\left(n\right) \notag\\
 & \quad +\log\left( \alpha+\frac{1}{2}\right)+\log \left( 1+\frac{1}{n}\right)\\
 &\le  \log(2\pi) +\left(\rme^{- \alpha}-1\right)\log\left(\alpha\right)+\frac{1}{2^n}\log\left(n\right) \notag\\
 & \quad +\log\left( \alpha+\frac{1}{2}\right)+\log \left( 1+\frac{1}{n}\right)\\
&= \log(2\pi) +\rme^{- \alpha}\log\left(\alpha\right)+\frac{1}{2^n}\log\left(n\right) \notag\\
&\quad +\log\left( 1+\frac{1}{2 \alpha}\right)+\log \left( 1+\frac{1}{n}\right)\\
&\le \log(2\pi) +\rme^{- \alpha}\log\left(\alpha\right)+\frac{1}{2^n}\log\left(n\right) \notag\\
& \quad +\log\left( 1+\frac{1}{2 }\right)+\log \left( 1+\frac{1}{n}\right)\\
&\le \log(2\pi) +1+\frac{1}{2^n}\log\left(n\right) +\log\left( 1+\frac{1}{2 }\right) \notag\\
& \quad +\log \left( 1+\frac{1}{n}\right), \label{eq:upper_on_g_n}
\end{align}
which is bounded in $n$.  Since \eqref{eq:upper_on_g_n} is strictly larger than \eqref{eq:upper_on_g_n_looser}, we can conclude the result in \eqref{eq:uniform_bound_on_g}.
\end{proof}

\section{Bounds on the Entropy of a Binomial Random Variable}\label{app:bound_binom_entropy}
First of all we need the following result.
\begin{lemma}\label{lem:expect_log_binomial}
    Let $P_{Y|X}(\cdot|x)$ be a Binomial pmf with $n$ trials and success probability $x$ per trial. Then,
    \begin{align}\label{eq:log_expect_binomial}
    &\expcnd{\mathbbm{1}(0<Y\le n)\log\left(\frac{Y}{n}\right)}{X=x} \notag\\
    & \quad \quad\ge   (1-(1-x)^n)\log(x)-1.
    \end{align}
\end{lemma}
\begin{proof}
    Inspired by the approach of~\cite[Appendix~B]{lapidoth2008capacity}, we bound the expectation as follows:
\begin{align}
    &\expcnd{\mathbbm{1}(0<Y\le n)\log\left(\frac{Y}{n}\right)}{X=x} \notag\\
    &=\expcnd{\mathbbm{1}(0<Y\le n)\log\left(x\right)}{X=x} \notag\\
    &\quad +\expcnd{\mathbbm{1}(0<Y\le n)\log\left(\frac{Y}{nx}\right)}{X=x} \\
    &=(1-(1-x)^n)\log\left(x\right) \notag\\
    & \quad +\expcnd{\mathbbm{1}(0<Y\le n)\log\left(\frac{Y}{nx}\right)}{X=x} \\
    &=(1-(1-x)^n)\log\left(x\right)+\sum_{y=1}^{n-1} P_{Y|X}(y|x) \log\left(\frac{y}{nx} \right) \\
    &\ge (1-(1-x)^n)\log\left(x\right)+\int_0^n P_{Y|X}(\lfloor y \rfloor|x) \log\left(\frac{y}{nx} \right) \rmd y \\
    &= (1-(1-x)^n)\log\left(x\right)+n\int_0^1 P_{Y|X}(\lfloor nt \rfloor|x) \log\left(\frac{t}{x} \right) \rmd t  \label{eq:exp_log_two_integrals}
\end{align}
where the inequality holds because $x\mapsto \log(x)$ is an increasing function and negative for $x \in (0,1)$.

Now introduce the continuous rv $Z$ with pdf $f_{Z}(z) = n P_{Y|X}(\lfloor nz \rfloor|x)$ for $z \in [0,1]$. Then, the integral of~\eqref{eq:exp_log_two_integrals} becomes:
\begin{align}
     &n\int_0^{1} P_{Y|X}(\lfloor nt \rfloor|x) \log\left(\frac{t}{x} \right) \rmd t \notag\\
     &= \int_0^{1} f_{Z}(t) \log\left(\frac{t}{x} \right) \rmd t \\
     &=\int_0^{x} f_{Z}(t) \log\left(\frac{t}{x} \right) \rmd t+\int_x^{1} f_{Z}(t) \log\left(\frac{t}{x} \right) \rmd t.
\end{align}
Let us now bound the two integrals separately. For the first integral, by integrating by parts we have
\begin{align}
     &\int_0^{x} f_{Z}(t) \log\left(\frac{t}{x} \right) \rmd t \notag\\
     &= \left[\Pr(Z \le t) \log\left(\frac{t}{x}\right) \right]_0^{x}-\int_0^{x} \Pr(Z \le t) \frac{1}{t} \rmd t \\
     &\ge-\int_0^{x} \int_0^{t}n P_{Y|X}(\lfloor nz \rfloor|x) \rmd z \frac{1}{t} \rmd t \label{eq:bound_on_cdf} \\
     &\ge -\int_0^{x} n P_{Y|X}(\lfloor nt \rfloor|x)   \rmd t \label{eq:use_increasing_part_of_binomial} \\
     &= -\int_0^{x} f_Z(t)   \rmd t \\
     &\ge -1
\end{align}
where 
in~\eqref{eq:use_increasing_part_of_binomial} we used that 
$ \int_0^{t}n P_{Y|X}(\lfloor nz \rfloor|x) \rmd z\le tn P_{Y|X}(\lfloor nt \rfloor|x)$  thanks to the following lemma and to $t \le x$:
\begin{lemma}
    Let $P_{Y|X}$ be a Binomial pmf. Then, $y\mapsto P_{Y|X}(y|x)$ is increasing for $y \le \lfloor (n+1)x \rfloor$, and decreasing for $y \ge \lceil (n+1)x \rceil$.
\end{lemma}
\begin{proof}
    From the ratio 
    \begin{align}
        \frac{P_{Y|X}(y|x)}{P_{Y|X}(y-1|x)} = \frac{n-y+1}{y} \frac{x}{1-x}
    \end{align}
    we see that the condition $P_{Y|X}(y|x) \ge P_{Y|X}(y-1|x)$ is satisfied for $y \le \lfloor (n+1)x \rfloor$.
\end{proof}
For the second integral, write
\begin{align}
    \int_{x}^1 f_Z(t) \log\left(\frac{t}{x} \right) \rmd t \ge 0.
\end{align}
Putting together the two results, we get the result in~\eqref{eq:log_expect_binomial}.
\end{proof}

We are now ready to give the main result of this appendix.
\begin{lemma}
    For $x \in [0,1]$, the entropy of a Binomial rv is bounded as follows
    \begin{equation}
        H(Y|X=x) \le \frac{1}{2}\log\left(2\pi \rme \left(n x(1-x)+\frac{1}{12}\right)\right),
    \end{equation}
    \begin{align}
        H(Y|X=x)&\ge (1-(1-x)^n-x^n)\frac{1}{2}\log\left( 2\pi n\right) \notag\\
        &+\frac{1}{2}(1-(1-x)^n)\log(x) \notag\\
        &+\frac{1}{2}(1-x^n)\log(1-x)-1.
    \end{align}
\end{lemma}
\begin{proof}
    For the upper bound, write
    \begin{align}
        H(Y|X=x) &= h(Y+U|X=x) \label{eq:apply_add_uniform} \\
        &\le \frac{1}{2}\log\left(2\pi \rme \left(n x(1-x)+\frac{1}{12}\right)\right),
    \end{align}
    where \eqref{eq:apply_add_uniform} follows from \cite[Lemma 17]{lapidoth2008capacity} with $U\sim {\cal U}[0,1]$ is independent of $Y$; and the last step follows from the Gaussian maximizes entropy principle. 

Next we prove the lower bound. First of all, compute
\begin{align}
    -H(Y|X=x) &= \expcnd{\log\left({n \choose Y}x^{Y}(1-x)^{n-Y}\right)}{X=x} \\
    &= \expcnd{\log{n \choose Y}}{X=x} + nx\log(x) \notag\\
    & \quad +n(1-x)\log(1-x) \\
    &\le \expcnd{\log{n \choose Y}}{X=x}-nH_2(x)
\end{align}
By using the bound $\binom{n}{Y}\le \sqrt{\frac{n}{2\pi Y(n-Y)}}\rme^{n H_2(\frac{Y}{n})}$ for $0<Y<n$ (see, e.g., \cite[Problem~5.8]{gallager1968information}), we can write:
\begin{align}
    &\expcnd{\log{n \choose Y}}{X=x} \notag\\
    &=\expcnd{\mathbbm{1}(0<Y<n)\log{n \choose Y}}{X=x} \\
    &\le (1-(1-x)^n-x^n)\frac{1}{2}\log\left( \frac{n}{2\pi}\right) \notag\\
    & \quad -\frac{1}{2}\expcnd{\mathbbm{1}(0<Y<n)\log(Y(n-Y))}{X=x} \notag\\
    & \quad +n\expcnd{H_2(\frac{Y}{n})}{X=x} \label{eq:bound_bin_coeff} \\
    &= -(1-(1-x)^n-x^n)\frac{1}{2}\log\left(2\pi n\right) \notag\\
    & \quad -\frac{1}{2}\expcnd{\mathbbm{1}(0<Y<n)\log\left(\frac{Y}{n}\frac{n-Y}{n}\right)}{X=x} \notag\\
    & \quad+n\expcnd{H_2 \left(\frac{Y}{n} \right)}{X=x} \\
     &= -(1-(1-x)^n-x^n)\frac{1}{2}\log\left( 2\pi n\right) \notag\\
     & \quad -\frac{1}{2}\expcnd{\mathbbm{1}(0<Y\le n)\log\left(\frac{Y}{n}\right)}{X=x} \nonumber\\
     &\quad-\frac{1}{2}\expcnd{\mathbbm{1}(0< Y\le n)\log\left(\frac{Y}{n}\right)}{X=1-x} \notag\\
     & \quad +n\expcnd{H_2(\frac{Y}{n})}{X=x}\label{eq:use_symmetry_bino_pmf} \\
     &\le -(1-(1-x)^n-x^n)\frac{1}{2}\log\left( 2\pi n\right) \notag\\
     & \quad -\frac{1}{2}\expcnd{\mathbbm{1}(0<Y\le n)\log\left(\frac{Y}{n}\right)}{X=x}\nonumber\\
     &\quad-\frac{1}{2}\expcnd{\mathbbm{1}(0< Y\le n)\log\left(\frac{Y}{n}\right)}{X=1-x}+nH_2(x), \label{eq:exp_log_bin_coeff}
\end{align}
where in~\eqref{eq:use_symmetry_bino_pmf} we used the channel symmetry $P_{Y|X}(y|x) = P_{Y|X}(n-y|1-x)$; and in the last step we used Jensen's inequality and $\expcnd{Y}{X=x}=nx$.

By using Lemma~\ref{lem:expect_log_binomial}, we have
\begin{align}
    &\expcnd{\mathbbm{1}(0<Y\le n)\log\left(\frac{Y}{n}\right)}{X=x} \notag\\
    & \quad \ge (1-(1-x)^n)\log(x)-1
\end{align}
and
\begin{align}
   & \expcnd{\mathbbm{1}(0<Y\le n)\log\left(\frac{Y}{n}\right)}{X=1-x} \notag\\
    & \quad \ge (1-x^n)\log(1-x)-1.
\end{align}
Therefore, we have
\begin{align}
    &\expcnd{\log{n \choose Y}}{X=x} \notag\\
    &\le -(1-(1-x)^n-x^n)\frac{1}{2}\log\left( 2\pi n\right) \notag\\
    & \quad -\frac{1}{2}(1-(1-x)^n)\log(x)-\frac{1}{2}(1-x^n)\log(1-x) \notag\\
    & \quad+1+nH_2(x)
\end{align}
and 
\begin{align}
    &-H(Y|X=x) \notag \\
    &\le -(1-(1-x)^n-x^n)\frac{1}{2}\log\left( 2\pi n\right) \notag\\
    & \quad -\frac{1}{2}(1-(1-x)^n)\log(x)-\frac{1}{2}(1-x^n)\log(1-x)+1.
\end{align}

\end{proof}

\end{appendices}

\end{document}